\PassOptionsToPackage{x11names,table,dvipsnames}{xcolor}
\documentclass[acmtog,nonacm]{acmart}
\acmSubmissionID{1430}
\usepackage{setspace,afterpage,fix-cm,multirow}
\usepackage{enumitem}
\setitemize{noitemsep,topsep=0pt,parsep=0pt,partopsep=0pt,leftmargin=*}

\usepackage{algorithm}
\usepackage{algorithmicx}
\usepackage{algpseudocode}
\usepackage{booktabs}
\usepackage{mathtools}
\usepackage{multirow}
\usepackage{tabstackengine}
\usepackage{subcaption}
\usepackage{siunitx}
\usepackage{bm}
\usepackage[normalem]{ulem}
\usepackage{xspace}

\algdef{SE}[DOWHILE]{Do}{doWhile}{\algorithmicdo}[1]{\algorithmicwhile\ #1}

\usepackage{wrapfig}

\mathtoolsset{centercolon}

\definecolor{dgreen}{rgb}{0.0, 0.5, 0.0}
\definecolor{dred}{rgb}{0.65, 0.16, 0.16}

\newcommand{\VO}{\emph{VelPro} }

\newcommand{\VPD}{Volumetric Progressive Dynamics\xspace} 

\usepackage{pifont}

\def\argmin{\mathop{\rm argmin}}

\acmJournal{TOG}
\setcopyright{acmlicensed}
\acmJournal{TOG}
\acmYear{2024} \acmVolume{43} \acmNumber{4} \acmArticle{104} \acmMonth{7}\acmDOI{10.1145/3658214}

\setcitestyle{square}

\setcitestyle{square}
\makeatletter

\renewcommand{\l@section}{\@dottedtocline{1}{1.5em}{2.6em}}
\renewcommand{\l@subsection}{\@dottedtocline{2}{4.0em}{3.6em}}
\renewcommand{\l@subsubsection}{\@dottedtocline{3}{7.4em}{4.5em}}
\makeatother

\definecolor{lightbluishgrey}{rgb}{0.76078,0.88235,0.92157}

\usepackage{xr}
\usepackage{amsmath, amsthm}

\newtheorem{lemma}{Lemma}

\begin{document}

\title{Progressing Level-of-Detail Animation of Volumetric Elastodynamics
}

%
%

\author{Jiayi Eris Zhang}
\affiliation{
  \institution{Stanford University}
  \country{USA}}
\affiliation{
  \institution{Adobe}
  \country{USA}}
\email{eriszhan@stanford.edu}

\author{Doug L. James}
\affiliation{
  \institution{Stanford University}
  \country{USA}}
\email{djames@cs.stanford.edu}

\author{Danny M. Kaufman}
\affiliation{
  \institution{Adobe}
  \country{USA}}
\email{dannykaufman@gmail.com}

\begin{abstract}

We extend the progressive dynamics model \cite{zhang2024progressive} from cloth and shell simulation to volumetric finite elements, enabling an efficient level-of-detail (LOD) animation-design pipeline with predictive coarse-resolution previews facilitating rapid iterative design for a final, to-be-generated, high-resolution animation of volumetric elastodynamics. This extension to volumetric domains poses significant new challenges, including the construction of suitable mesh hierarchies and the definition of effective prolongation operators for codimension-0 progressive dynamics.

To address these challenges, we propose a practical method for defining multiresolution hierarchies and, more importantly, introduce a simple yet effective topology-aware algorithm for constructing prolongation operators between overlapping (but not necessarily conforming) volumetric meshes. Our key insight is a boundary binding strategy that enables the computation of barycentric coordinates, allowing several off-the-shelf interpolants—such as standard barycentric coordinates, Biharmonic Coordinates \cite{wang2015linear}, and Phong Deformation \cite{James2021}—to serve as ``plug-and-play'' components for prolongation with minimal modification. We show that our progressive volumetric simulation framework achieves high-fidelity matching LOD animation across resolutions including challenging dynamics with high speeds, large deformations, and frictional contact.

\end{abstract}

\begin{teaserfigure}
\centering
\includegraphics[width=1\linewidth]{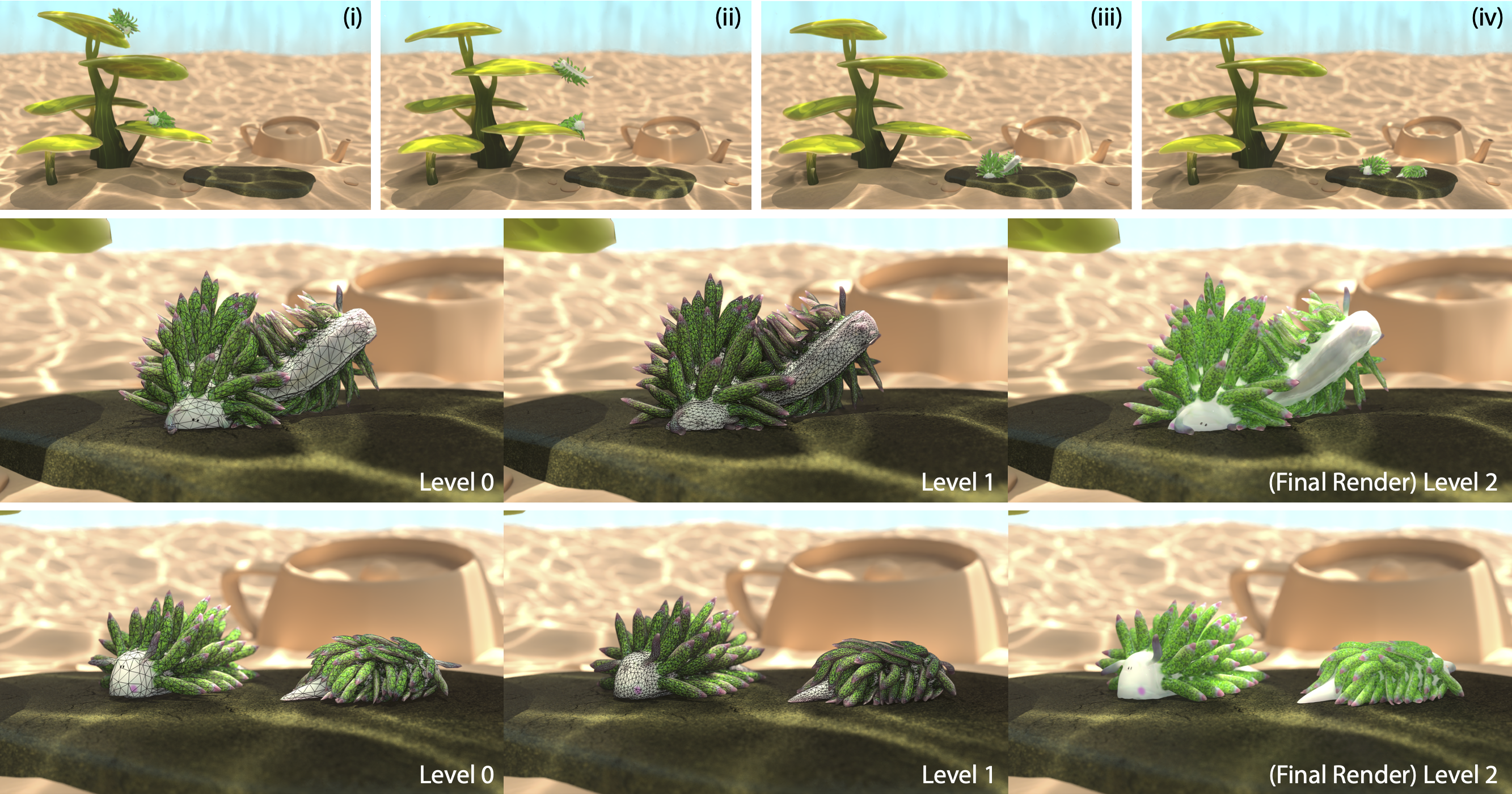}
   \vspace{-0.5cm}
   \caption{{\bf Two Squishy Cartoonish Leaf Sheep on a Ballistic Adventure:} (Top) An animator explores a tricky shot where (left to right) two Leaf Sheep on a plant (i) experience an initial disturbance, (ii) fall and roll off adjacent leaves, (iii) impact each other as they hit a stone, then (iv) somehow land upright ready for grazing. The design is tediously constructed by manually adjusting initial conditions using a fast, coarse (level 0) tetrahedral volume simulation (10K tetrahedra), which is then refined using progressive volumetric dynamics to obtain finer level 1 (73K tetrahedra) and level 2 (400K tetrahedra) simulation models with predictably consistent motions. Zoomed-in shots of frame iii (Middle) and frame iv (Bottom) show that the refined mesh deformations retain the character of the fast level-0 preview simulation, while enhancing detail of the deformations and contact interactions. Since this coarse level-0 simulation is more than 100 times faster than the final level-2 simulation, progressive volumetric dynamics greatly improves the flexibility and efficiency of previewing various design options.}
    \label{fig:teaser}
\end{teaserfigure}

\maketitle

\section{Introduction}

\begin{figure}[t!]
  \centering
  \includegraphics[width=\linewidth,keepaspectratio]{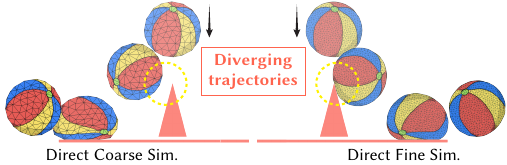} 
  \caption{{\bf Diverging Trajectories:} Direct simulations of a ball bouncing on a spike bounce one way (Left) at coarse resolution, but another way (Right) at fine resolution due to subtle mesh differences. In contrast, \VPD avoids trajectories diverging at different resolutions.
    \label{fig:beachball-direct}}
\end{figure}

Volumetric objects are ubiquitous, and volumetric elastodynamics plays a central role in applications such as visual effects, video games, virtual reality, soft body animation, and even stylized 2D cartoon animation. These scenarios often prioritize visual plausibility, artistic control, and interactive performance over strict physical accuracy. By modeling the behavior of elastic materials, volumetric dynamics helps bring deformable objects and their complex interactions to life in a visually compelling way. However, achieving high-quality volumetric animation often requires high-resolution discretizations with a large number of degrees of freedom, leading to substantial computational costs. This burden poses a significant challenge for iterative creative workflows, where designers and artists need to explore multiple design variations and rely on rapid visual feedback.

To accelerate the iterative design process, recent work \cite{zhang2024progressive} introduce the Progressive Dynamics framework that addresses the long-standing challenge of enabling rapid iterative design for high-fidelity cloth and shell animation. This framework facilitates efficient modeling and animation workflows by generating predictive coarse previews that progressively refine into high-resolution results through a coarse-to-fine level-of-detail (LOD) approach. While this technique has proven highly effective for codimension-1 objects such as cloth and shells, extending it to volumetric domains presents fundamental challenges. In the shell case, Zhang et al. \shortcite{zhang2023progressive} construct a top-down shell mesh hierarchy by applying recursive edge-collapse surface decimation \cite{garland1997surface}, and track bijective mappings between local patches during decimation \cite{Lee1998,aksoylu2005multilevel,Liu:2021:SMIP}. These mappings are then composed to define a custom prolongation operator tailored for curved shell geometry and rest-shape preservation.

In contrast, for volumetric meshes, it is unclear how to construct effective hierarchies, and existing tetrahedral decimation methods~\cite{cignoni2000simplification,staadt1998progressive,renze1996generalized,trotts1999simplification,danovaro2002multiresolution} naturally do not yield such local mappings. Furthermore, even if such a hierarchy is provided, defining prolongation operators without relying on tracked correspondences remains another big open question, particularly since volumetric levels often exhibit not only non-conforming boundaries, but also partially overlapping or mismatched interior regions.

To address these challenges, we first propose a practical pipeline for generating multiresolution tetrahedral mesh hierarchies by combining standard surface mesh decimation algorithms \cite{garland1997surface, 10.1111:cgf.13933} with tetrahedralization methods \cite{hang2015tetgen, hu2018tetrahedral}. In general, we assume the resulting hierarchy consists of tetrahedral meshes that approximate the same domain but may have non-conforming boundaries.

Building on this, we introduce a simple yet effective topology-aware algorithm for constructing linear prolongation operators between these overlapping (but not necessarily conforming) volumetric meshes. Treating each coarse mesh as the embedding shape, our key insight is a boundary-binding strategy that enables the computation of barycentric coordinates (possibly negative) for the fine mesh, via extrapolation, and vice versa. This construction, in turn, enables several off-the-shelf interpolants—such as standard barycentric coordinates, Biharmonic Coordinates \cite{wang2015linear}, and Phong Deformation \cite{James2021}—to serve as ``plug-and-play'' components for prolongation with minimal modification. We show that all of these interpolants are fully compatible with our \VPD framework. With the key feature of consistent preview and refinement across LOD, our \VPD method may also be interpreted as a kind of physics-based relaxation tool for post-processing embedded simulations using the aforementioned interpolation or skinning methods. It enhances upsampled results by making them more physics-aware, thereby enabling the recovery of additional physical details, such as higher-frequency motions and improved contact resolution.

We present extensive evaluations demonstrating that our proposed \VPD framework generates high-fidelity, LOD animations across resolutions and handles challenging elastodynamic scenarios, including high-speed motion, large deformations, and frictional contact, in both 2D and 3D. For quantitative evaluation, we adopt the temporal continuity and geometric consistency metrics introduced by Zhang et al. \shortcite{zhang2025progressive} to assess animation quality across levels. Additionally, we incorporate their proposed quadratic penalty term to support optional user-controlled balancing between geometric consistency and enrichment in our LOD animation results.
\section{Related Work}

\subsection{Progressive Simulation}

\begin{figure*}[t!]
  \centering
  \includegraphics[width=\linewidth,keepaspectratio]{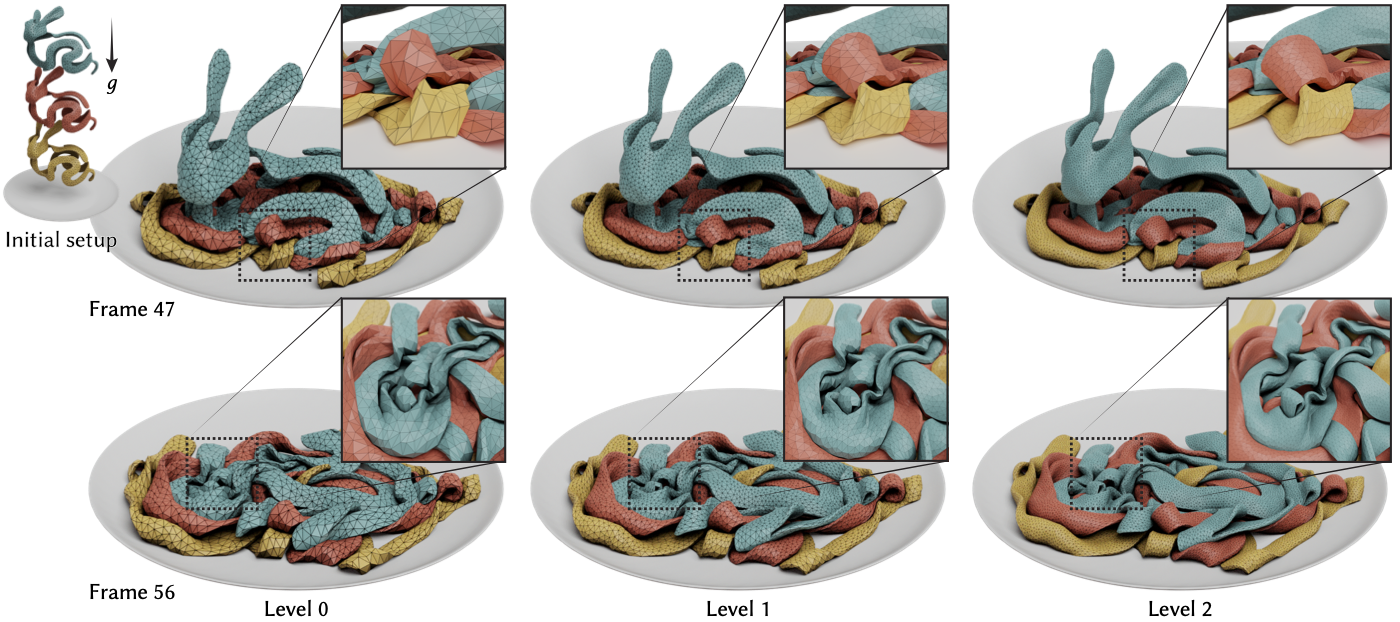} 
  \caption{{\bf Bunny noodles:} We show that \VPD works for shell-like volumetric materials, where (Left) the coarsest level-0 simulation mesh is only one element thick. Increasing simulation resolutions (Left to Right) demonstrate consistent and increasingly detailed contact resolution. Yum! }
  \label{fig:bunny-noodles}
\end{figure*}

The recently introduced progressive simulation framework \cite{10.1145/3550454.3555510, zhang2023progressive, zhang2024progressive, zhang2025progressive} provides an effective solution to the long-standing challenge of enabling rapid iterative design for high-fidelity cloth and shell simulations. It facilitates efficient modeling and animation workflows by generating predictive coarse previews that progressively refine to high-resolution finest-level results via a coarse-to-fine LOD workflow. Although these methods have shown success in both quasistatic \cite{10.1145/3550454.3555510, zhang2023progressive} and dynamic \cite{zhang2024progressive} cloth and shell simulations, their focus remains limited to codimensional domains.

Our work extends progressive dynamics \cite{zhang2024progressive, zhang2025progressive} to volumetric finite element simulations, a critical area for accurately capturing the behavior of a wider range of objects in the real world. By extending progressive dynamics to volumetric domains, we expand its applicability and enable more comprehensive physics-based simulations that support diverse scenarios beyond cloth and shells in a unified framework.

\subsection{Multiresolution Volumetric Mesh Hierarchy}
Multiresolution mesh hierarchies for volumetric meshes have been extensively studied to enable efficient storage, processing, and simulation of complex 3D structures. Early work on hierarchical tetrahedral decompositions \cite{Debunne2001} introduced adaptive refinement strategies for physics-based simulations, which enable locally increased resolution in areas of high deformation while maintaining computational efficiency. Progressive volumetric representations, such as tetrahedral mesh generation techniques \cite{staadt1998progressive, Molino2003}, extend the concept of progressive surface meshes \cite{Hoppe1996} to 3D domains, supporting dynamic LOD control through tetrahedral collapse and vertex split operations. In fluid simulation, octree-based adaptive methods \cite{Losasso2004} enable dynamic refinement of volumetric discretizations, which allows fine-scale resolution in regions with complex fluid behavior while maintaining coarser representations elsewhere. This adaptivity significantly improves performance in large-scale simulations without compromising detail. Additionally, learning-based approaches have begun to integrate hierarchical volumetric structures with neural architectures, as seen in data-driven tetrahedral mesh reconstruction methods \cite{wang2017cnn, Gao2020}. Despite these advancements, designing suitable and scalable multiresolution volumetric hierarchies that effectively balance geometric fidelity, computational efficiency, and numerical robustness still requires extensive customization for specific applications and remains a challenging problem.

\vspace{-3mm}

\subsection{Embedded Simulation and Physics Skinning}
Embedded simulation using both linear- and higher-order elements has become a widely adopted technique for reducing computational costs in physics-based animation, particularly in applications involving complex deformable objects \cite{faloutsos_dynamic_1997,molino_virtual_2004,Capell:2002:ISDDD,muller_interactive_2004, rivers_fastlsm_2007, zhu2010efficient, McAdams:2011:MGCharacters}. Linear interpolation, while computationally efficient, often introduces visual artifacts such as deformation gradient discontinuities, faceting artifacts, self-intersections, and inconsistent contact behaviors. Methods such as subdivision surfaces \cite{derose1998subdivision}, multiresolution displacement mapping \cite{james2003multiresolution}, and recent advances like Phong Deformation \cite{James2021} mitigate these issues to some extent by introducing smoother transitions or higher-fidelity surface details. While complementary, these approaches fundamentally fail to incorporate sufficient fine-scale simulation information into the upsampling process, leading to a lack of detailed nonlinearities in the resulting deformation. While $C^1$ and higher-order interpolants provide greater smoothness and nonlinear detail, they remain fundamentally constrained by the limited information available in the coarse embedded domain, making them inadequate for resolving complex effects such as contact.

Popular skinning-based methods that utilize generalized barycentric coordinates can smoothly deform embedded geometry by treating coarse shapes as control rigs \cite{Jacobson2011, joshi2007harmonic, Ju2005, floater2003mean, wang2015linear}. While effective for tasks such as character articulation and cage-based deformation, these techniques are often designed for arbitrary control geometries (``cages'') and are not typically applied to dense simulation meshes (although see \cite{ruan2024minnie,xian_scalable_2019}). More critically, similar to linear and higher-order interpolation schemes, they struggle to incorporate fine-scale physical details or resolve contacts, leading to artifacts like gaps, self-intersections, and inconsistent deformation behaviors. 

Our method addresses these limitations by enabling the use of any embedding or skinning interpolant as a prolongation operator within our progressive volumetric simulation framework. This facilitates the seamless transfer of velocities and other data across levels, integrating them into a contact-aware system. In some sense, our approach can be viewed as a relaxation technique that post-processes embedded simulations, enriching them with fine-scale physical details and resolving contact inconsistencies.

\subsection{Volumetric Mapping}
Maps between 3D objects and volumetric parameterizations over 3D domains are critical tools in numerous geometric processing and analysis tasks, ranging from deformation transfer to physics-based simulations. While robust solutions to the construction of bijective 2D maps for triangle meshes exist and have been well studied \cite{Tutte1963}, their extension to 3D volumetric mapping remains comparatively underexplored.

Existing methods for volumetric mapping generally fall into two categories: optimization-based approaches and combinatorial methods. Optimization-based approaches \cite{kovalsky2014controlling, aigerman2013injective} often aim to achieve injectivity, low distortion, or harmonicity. However, these techniques typically rely on solving nonlinear and nonconvex optimization problems, which are computationally expensive and prone to local minima. On the other hand, combinatorial approaches \cite{campen2016bijective, cherchi2023volmap, hinderink2023galaxy, hinderink2024bijective} leverage discrete formulations to guarantee bijectivity and robustness. Unfortunately, these methods are often restricted to specific topologies, such as ball-like or star-shaped domains, which limits their applicability to general 3D input geometries with arbitrary topology needed in physics-based simulation.


\subsection{Multigrid Prolongation}

Our volumetric prolongation operator is related to the techniques used for multigrid volumetric solvers (see \cite{10.1145/3528223.3530109} for an excellent summary). For efficiency and practicality, multigrid methods often exploit regular grids for efficient discretization and interpolation and are standard in graphics and engineering \cite{bolz_sparse_2003,trottenberg_multigrid_2001}.
In computer animation, \citet{zhu2010efficient} used trilinear interpolation to perform prolongation on uniform 3D grids suitable for finite-difference discretizations of elasticity for embedded character deformation. In a related work \cite{McAdams:2011:MGCharacters}, multigrid instability issues related to extrapolation outside the coarse-level domains were reduced using gradient-based prolongation techniques.

Recent multigrid methods for deformable models in graphics \cite{xian_scalable_2019,ruan2024minnie} have leveraged ``skinning space coordinates'' \cite{10.1145/3197517.3201387,10.1145/2185520.2185573} to interpolate coarse information to fine scales on a grid hierarchy. For general hierarchies and unstructured problems, algebraic multigrid approaches can also be used \cite{brandt_algebraic_1986,10.1145/3528223.3530109}, leading to their own algebraically inferred prolongation and restriction operators.

\citet{Georgii:2006:MG4RTDefo} use an unstructured nonnested hierarchy for tetrahedral FEM models, and employ barycentric interpolation to construct their prolongation operator, both for nodes contained within a coarser-level element, but also for nodes outside the coarse mesh via extrapolation using the nearest tetrahedral element. A similar approach is also used to construct a two-level tetrahedral mesh hierarchy for a numerical coarsening scheme in \cite{kharevych2009numerical}. Our method is most closely related to these approaches; however, we show that extrapolation can lead to binding-related deformation artifacts, propose a solution (see Figure~\ref{fig:binding-viz}) and bound extrapolation effects, and generalize to nonbarycentric interpolants (Biharmonic Coordinates \cite{wang2015linear} and Phong Deformation \cite{James2021}).

\subsection{Adaptive and Subspace Contact}
Adaptive and subspace methods have been widely used to efficiently simulate elastodynamics with contact by adaptively reducing the simulation to smaller subspaces or fewer degrees of freedom \cite{grinspun2002charms, capell2002multiresolution, ferguson2023timestep, teng2015subspace, kim2009skipping, zheng2011toward}. These approaches leverage model reduction techniques and adaptive basis updates to accelerate physics-based animation while maintaining accuracy in specific scenarios.

However, while these methods excel in computational efficiency in certain cases, reduced subspace methods can restrict the range of deformations; and contact processing that occurs at the fine scale \cite{trusty2024trading} can limit simulation performance. These methods can lack support for previewing and progressive refinement, which are essential for offline computations that require high-resolution simulations for a final high-fidelity output. More critically, adaptive schemes that modify degrees of freedom on the fly can struggle with consistency when transitioning between different resolutions. As such, these methods are fundamentally different from the goals of our work, which are to provide a progressive framework that balances efficiency and consistent uniform refinement in contact-rich simulations.
\section{Method}

\subsection{Volumetric Hierarchy}
\label{sec:vol-hierarchy}

\subsubsection{Considerations for Progressive Dynamics}
We begin by describing the construction of the volumetric mesh hierarchies that Progressive Dynamics will apply in its multilevel solver. Our objective is to create a hierarchy of volumetric meshes (triangle meshes in 2D or tetrahedral meshes in 3D) that represent the same object in space at different levels of detail. Similarly to Zhang et al. \shortcite{zhang2023progressive, zhang2024progressive}, we assume that the hierarchy is constructed starting with an input fine-resolution mesh $\mathcal{M}_L$, which consists of $|\mathcal{V}_L|$ vertices and $|\mathcal{T}_L|$ elements.  The resulting mesh hierarchy is indexed by resolution using the subscript $l \in \{0,1,\dots,L\}$ as $\mathcal{M}_L, \mathcal{M}_{L-1}, \dots, \mathcal{M}_0$. At any level $l$, the undeformed (rest) positions of the mesh nodes are denoted by $\bar{x}_l \in \mathbb{R}^{3n_l}$, and the deformed positions by $x_l \in \mathbb{R}^{3n_l}$, where $n_l$ is the number of vertices at level $l$. As an analogy, in the shell case described by Zhang et al. \shortcite{zhang2023progressive, zhang2024progressive}, the shell mesh hierarchy is constructed in a top-down manner by recursively applying edge-collapse decimation \cite{garland1997surface} to a high-resolution input mesh. Most importantly, bijective mappings between local patches can be tracked during this process \cite{Lee1998,aksoylu2005multilevel,Liu:2021:SMIP}, and these mappings are then composed to define a custom prolongation operator for curved shell geometry and rest-shape preservation.

However, extension of this approach to volumetric domains is nontrivial. While various tetrahedral decimation methods exist \cite{cignoni2000simplification,staadt1998progressive,renze1996generalized,trotts1999simplification,danovaro2002multiresolution}, they do not naturally generate the local correspondences required to construct prolongation operators. As a result, the hierarchy-and-prolongation strategy used in the shell case cannot be directly applied to volumetric meshes. Due to this limitation, we instead construct the hierarchy by combining robust, off-the-shelf algorithms for surface decimation and volumetric tetrahedralization. Unlike the shell case, our prolongation operators are then defined independently of any tracked correspondences from decimation, as detailed in Section~\ref{sec:prolongation}.

\begin{figure}[t!]
  \centering
  \includegraphics[width=\linewidth,keepaspectratio]{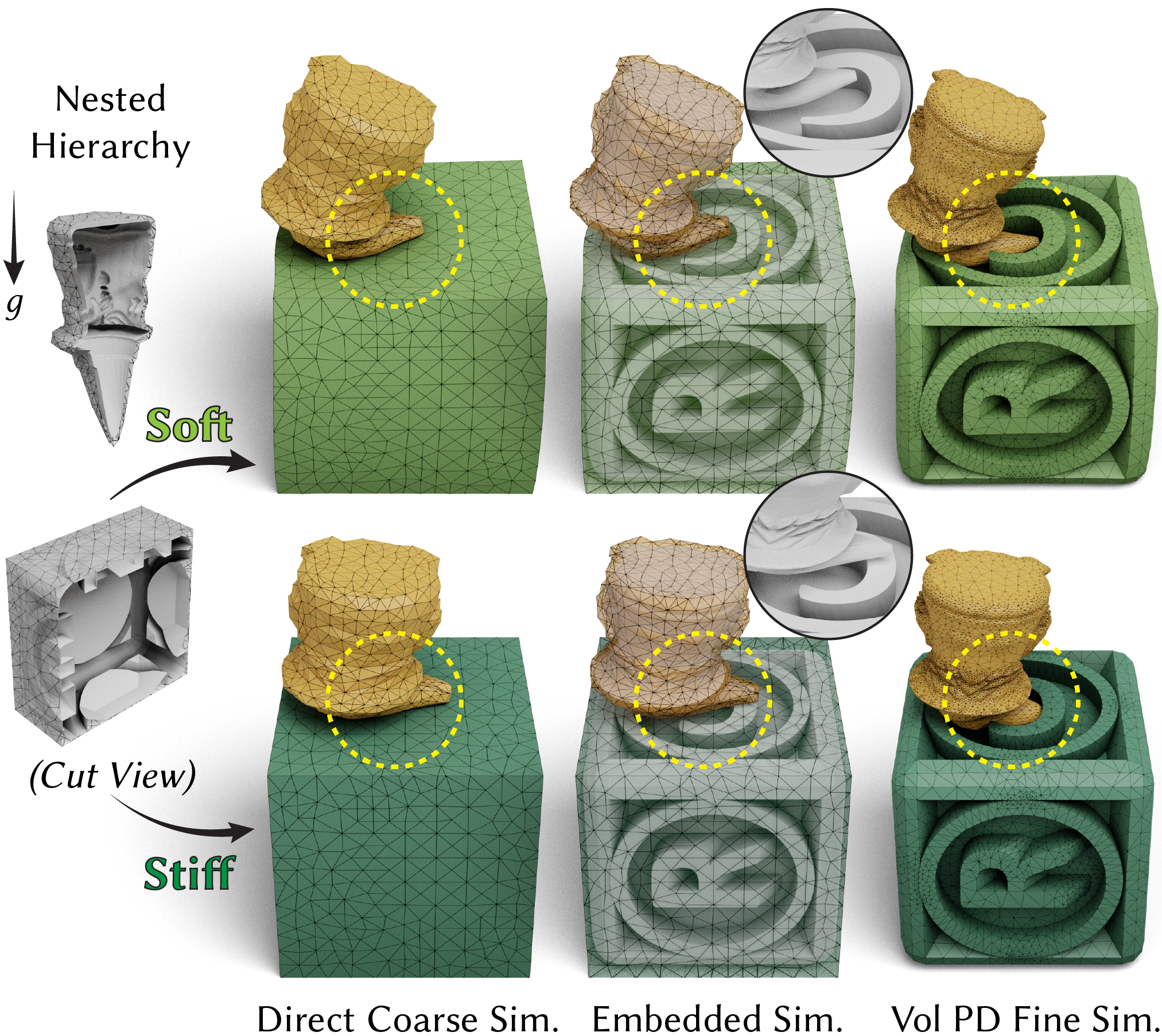} 
  \caption{{\bf Contact geometry resolution and conformity:} We show that \VPD effectively handles contact resolution and geometric conformity for both (Top) soft and (Bottom)  materials, overcoming inaccurate contact handling of embedded simulation that can result in unnatural gaps. (Left) Cut-away view of the embedded hierarchy, showing only the boundary surface mesh of the underlying tetrahedral meshes.}
  \label{fig:dice-and-wizard}
\end{figure}

\subsubsection{Construction of Mesh Hierarchy}
Our approach proceeds in two stages. First, we extract the surface of the input tetrahedral mesh $\mathcal{M}_L$, denoted by $\mathcal{S}_L$, and apply the aforementioned surface mesh decimation algorithm (or other variants that produce high-quality output) to generate a hierarchy of surface meshes $\mathcal{S}_L, \mathcal{S}_{L-1}, \dots, \mathcal{S}_0$. Then, for each surface mesh, we apply robust tetrahedralization tools \cite{hang2015tetgen,hu2018tetrahedral, hu2020fast}, to generate the corresponding volumetric meshes $\mathcal{M}_L, \mathcal{M}_{L-1}, \dots, \mathcal{M}_0$. While TetGen \cite{hang2015tetgen} exactly preserves the boundary of the input surface, we use TetWild \cite{hu2018tetrahedral} as it offers better control over the quality of the output tetrahedral mesh, although it may not preserve the boundary. As a result, the generated tetrahedral meshes may not exactly conform to the original surfaces $\mathcal{S}_l$, which is acceptable for our purposes, since the surfaces already have different boundaries, and strict boundary preservation is not required for subsequent prolongation construction. This procedure provides a practical and systematic way to construct a hierarchy for Progressive Dynamics. In practice, any user-provided tetrahedral mesh hierarchy generated by the aforementioned multiresolution tetrahedralization algorithms can also be used, as long as the meshes approximate the same domain, maintain good element quality (e.g., no degenerate tetrahedra), and do not differ significantly from one another.

\subsubsection{Volumetric Elastodynamics}
We adopt the IPC formulation \cite{Li2020IPC} for volumetric elastodynamics, where elastic deformation, contact, and friction are modeled through a total potential energy composed of three terms: elastic energy ($\Psi$), contact barrier energy ($B$), and frictional dissipation ($D$). 
These components together define the total potential energy at level $l$ as\footnote{The elastic energy $\Psi$ can be modeled using a variety of common volumetric materials, such as St.~Venant–Kirchhoff, Neo-Hookean, or co-rotational elasticity. Contact and friction are handled using IPC barrier potentials \cite{Li2020IPC}.} $E_l = \Psi_l + B_l + D_l.$ With these energies defined at each level, we recall that the \emph{direct} forward time step for frictional volumetric dynamics can be formulated variationally, as the minimization of an incremental potential (IP) \cite{kane1999finite,Li2020IPC}, which combines the level-specific energy $E_l$ with an inertia term. This formulation supports a broad class of implicit time-integration methods. Accordingly, at each discretization level $l$, an implicit Euler time step is computed by minimizing the following energy:
\begin{equation}
\label{eq:IE}
    x_l^{t+1} = \argmin_x \frac{1}{2 h^2} ||x-\tilde{x}_l^t||^2_{M_l} +  E_l(x),
\end{equation}
where $h$ is the time step size, $M_l$ is the mass matrix for level $l$, and $E_l$ is the total potential energy, also  incorporating contributions from body and external forces. Here, $\tilde{x}_l^t = x_l^t + h v_l^t$ denotes the implicit Euler velocity update, where $v_l^t$ is the velocity at the current step $t$. 

\subsection{Progressive Advancement}
\label{sec:prog-advancement}
Directly simulating each level independently (using \eqref{eq:IE}) can cause animation trajectories to diverge across levels due to the lack of inter-level communication. Consequently, coarse-level animations do not provide faithful previews of yet-to-be-run fine-level animations. Progressive Dynamics \cite{zhang2024progressive,zhang2025progressive} addressed these challenges by introducing two novel strategies (subspace proxy energies and prolongated velocity updates) to jointly enable high-quality coarse-level previews of cloth and shell dynamics, which could be progressively refined to produce a final high-resolution animation. Both ingredients leveraged the construction of a prolongation operator $P_{l+1}^l(\cdot)$, which maps quantities from a coarser mesh $\mathcal{M}_l$ to a mesh $\mathcal{M}_{l+1}$ with the next finer resolution, and is established during the construction of the mesh hierarchy. We now explain our volumetric adaptation of these elements. 

\subsubsection{Proxy Energy Model (A free lunch)}
In prior work, the subspace proxy energy was an effective way to mitigate well-known membrane-locking artifacts that plague low-resolution shell simulations, and was essential for allowing coarse-level shell results to provide meaningful previews. The core idea is to reparameterize the elasticity term via $\Psi_L(P^l(x_l))$, allowing coarse-level subspaces to evaluate fine-scale elastic energies $\Psi_L$ through the prolongation operator $P^l(x_l)$, where $P^l$ refers to prolongation from level $l$ to the finest mesh at level $L$. 
Unfortunately, these fine-scale subspace energy evaluations can be more computationally expensive for detailed volumetric simulations, and, therefore, frustrate the goal of fast coarse previews. 

While it is often said that there is no free lunch, somewhat surprisingly, we observed that \VPD produces visually plausible results even without applying subspace elasticity energy parameterization at coarser levels ($l < L$). In other words, direct coarse volumetric simulations already yield good quality results that provide reasonable previews for a wide range of challenging volumetric examples (as demonstrated in Section~\ref{sec:evaluation}). Future work will investigate advanced homogenization techniques to further mitigate artifacts such as shear locking when using linear finite elements, as well as fast schemes for proxy energy evaluation, thereby further strengthening the capabilities of our \VPD framework.

\subsubsection{Velocity Update}
Another key contribution of \cite{zhang2024progressive} is the introduction of a novel velocity update strategy that combines progressive spatial refinement (as introduced by Zhang et al. \shortcite{10.1145/3550454.3555510, zhang2023progressive}) and forward dynamics over time to truly enable frame-by-frame consistent shell dynamics preview. This approach represents the state of the system on a multiresolution spatial grid, where spatial positions $x_l^t$ and velocities $v_l^t$ are defined at each point of the grid $(t,l)$, corresponding to the time step $t \in \{0, 1, \dots, N\}$ and the resolution level $l \in \{0, 1, \dots, L\}$. For the coarsest level $(l=0)$, we similarly compute forward time stepping by solving the coarse-level IP problem (Equation~\ref{eq:IE}) for the entire time interval, using the standard momentum update $\tilde{x}_l^t = x_l^t + h v_l^t$, where $v_l^t = (x_l^t - x_l^{t-1}) / h$ under the implicit Euler method. This process generates the preview state $(x^t_0, v^t_0)$ for all $t \in \{0, 1, \dots, N\}$, which forms the base row of the Progressive Dynamics solution grid.

For finer levels $l > 0$, velocity updates need special attention with important trade-offs for different options. Zhang et al. \shortcite{zhang2025progressive} recently identified 
the \emph{\textbf{VelPro}} velocity update as a stable and high quality update option 
with a prolonged diagonal update at grid points $(t+1, l+1)$ of
\begin{align}
\begin{split}
    \hat{x}_{l+1}^{t}
&= x_{l+1}^{t} + h \, \big(V_{l+1}^l(x_l^t)\big) v_l^t  \\
&= x_{l+1}^{t} + \big(V_{l+1}^l(x_l^t)\big) (x_l^t - x_l^{t-1}),
\end{split}\label{eq:r2}
\end{align}
where $V_{l+1}^l(x) = \nabla P_{l+1}^l(x)$ is the velocity prolongator, i.e., the Jacobian of the position prolongation operator and $v_l^t = (x_l^t - x_l^{t-1}) / h$. This step serves as the key ingredient to promote smooth progression across levels while maintaining frame-wise consistent dynamics. Finally, time-step advancement for finer levels at grid points $(t+1,l+1)$ is achieved by solving the progressive IP problem:
\begin{align}
\label{eq:prog-solve} 
\begin{split}
	x_{l+1}^{t+1} = \argmin_{x} \frac{1}{2 h^2} \|x 
	- \hat{x}^{t}_{l+1} \|^2_{M_{l+1}} + E_{l+1}(x).
\end{split}
\end{align}

\subsubsection{Consistency Penalty (Optional)}
Moreover, Zhang et al.~\shortcite{zhang2025progressive} observe that while the \VO integration method generally produces stable, continuous, and consistent results, per-frame geometric consistency may break in certain extreme cases, such as using large time steps together with a large number of resolution levels (e.g., $h = 0.04$s and $L = 8$). To address this, they proposed adding a small quadratic consistency bias term to the total energy, which enables user-controlled trade-offs between consistency and enrichment. The modified potential energy at level $l$ is then defined as:
\begin{align}
\label{eq:quadric-penalty}
W_l(x) = E_l(x) + w ||x - P_{l}^{l-1} x^{t+1}_{l-1}||_{M_{l}}^2,
\end{align}
where $w \geq 0$ is the consistency penalty weight, and the mass-weighted quadratic penalty encourages agreement between the current level’s solution and the prolonged solution from the next coarser level, $P_{l}^{l-1} x^{t+1}_{l-1}$. Given the inherent properties of Progressive Dynamics, even a very light penalty term is often sufficient to substantially enhance consistency. However, to be clear, {\em we achieve visually plausible and consistent LOD results without using the penalty term for all examples in this paper} (modulo Figure~\ref{fig:increasing-weights}); i.e., we set $w = 0$ throughout. However, we acknowledge that there may be cases where users find the penalty helpful to further improve consistency. We detail the influence of this penalty in the Limitations Section (\S\ref{sec:conclusion}) and illustrate its effects in Figures~\ref{fig:increasing-weights} and \ref{fig:consistency-error}.

\subsection{Prolongation of Volumetric Data}
\label{sec:prolongation}
\setlength{\columnsep}{0.5em}
\setlength{\intextsep}{0em}
\begin{wrapfigure}{r}{65pt}
\includegraphics[width=65pt]{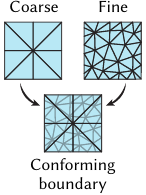}
\end{wrapfigure}

\subsubsection{Motivation} With these core components of our volumetric framework in place, an important remaining question is how to effectively construct prolongation operators to map quantities between nonconforming tetrahedral meshes at different levels. Specifically, as in the shell case, we aim to define a (preferably) linear prolongation operator $P_{l+1}^l$ that maps quantities from the tetrahedral mesh $\mathcal{M}_l$ to $\mathcal{M}_{l+1}$. We start by considering the simplest case: a two-level hierarchy \emph{with} conforming boundaries. In this idealized scenario, where the meshes are perfectly aligned, the prolongation operator can be easily constructed using barycentric coordinates \cite{schneider2017theory} or an $L^2$ projection operator \cite{leger2014updated, vavourakis2013assessment}. 

However, this boundary-conforming assumption is too restrictive for our framework. For complex geometries found in practical applications, with detailed surfaces and intricate structures, enforcing such constraints when constructing the hierarchy is both impractical and computationally inefficient. Hence, to build prolongation operators for the general case for shapes with nonconforming boundaries, we propose a practical extrapolation algorithm to address this problem. 

In the following, we first detail the algorithm using simple barycentric coordinates as an example; a setting similar to \cite{Georgii:2006:MG4RTDefo} but with extensions for robust vertex-element binding. However, the choice of prolongation operator is in fact flexible: we will show that a range of off-the-shelf interpolants commonly used in embedded simulation are compatible with our \VPD framework and can be used as ``plug-and-play'' options with minimal modification. In Section~\ref{sec:alternative-interpolants} we detail our framework and introduce additional interpolants based on Biharmonic Coordinates (\ref{sec:prolongBiharmonic}) and Phong Deformation (\ref{sec:prolongPhong}).

\subsubsection{Robust Boundary Extrapolation (Barycentric Case)}

\begin{figure}[t!]
  \centering
  \includegraphics[width=\linewidth,keepaspectratio]{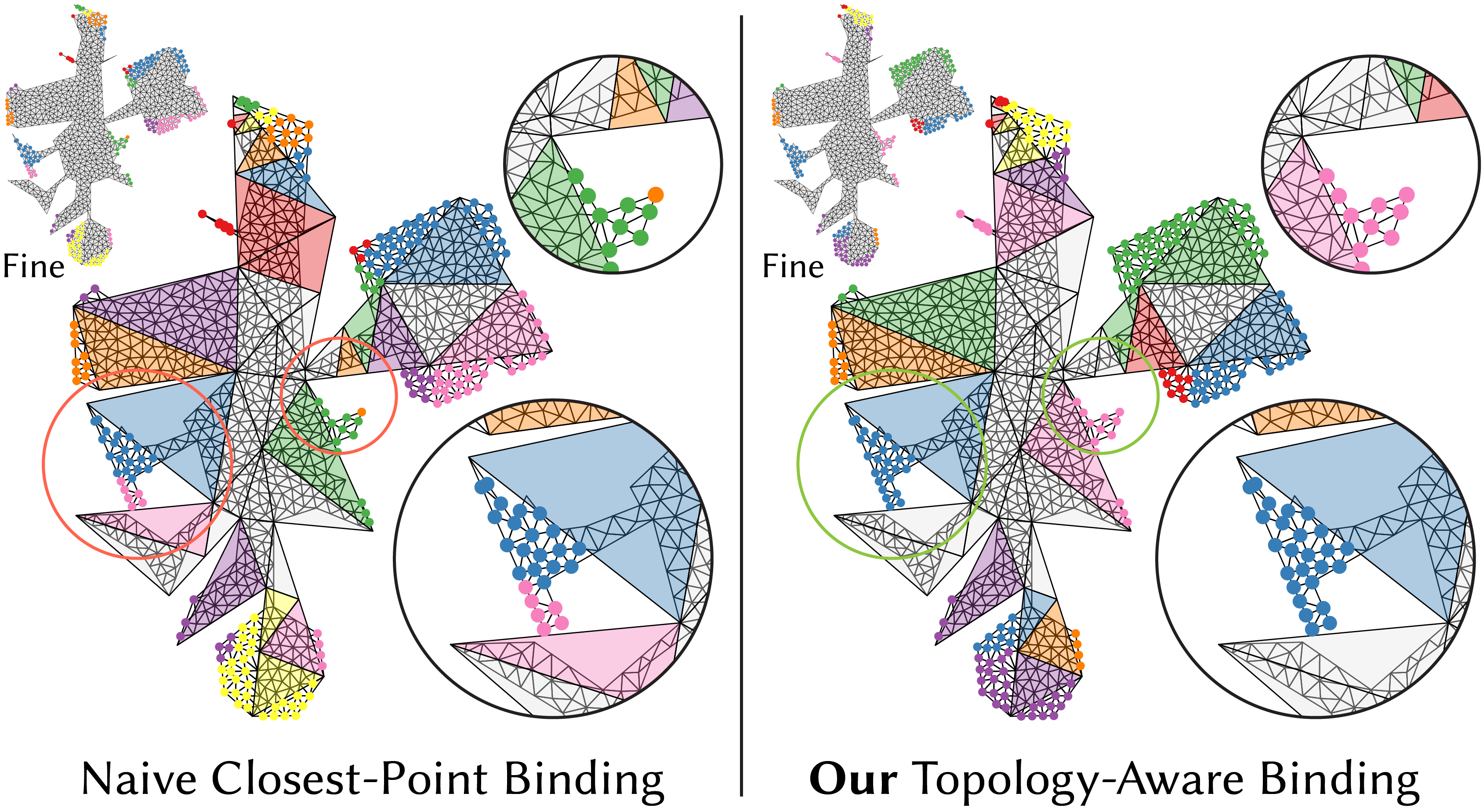} 
  \caption{{\bf Robust vertex-element binding:} We visualize the binding between a fine mesh and coarse elements, where dots of the same color represent fine vertices bound to the same coarse element. The na\"ive approach of closest-point binding using Euclidean distance often binds fine vertices to geodesically distant coarse elements, leading to bindings that are not semantically meaningful and introducing simulation artifacts. In contrast, our improved binding method ensures robust and topologically coherent connections.}
  \label{fig:binding-viz}
\end{figure}

To address the challenge of non-conforming boundaries, we propose a simple yet highly effective barycentric extrapolation method for constructing prolongations that are well-suited for volumes. Similar techniques have been widely applied for embedded simulation, where enforcing a perfectly embedded hierarchy is often too restrictive or impractical\ \cite{kharevych2009numerical, James2021, Georgii:2006:MG4RTDefo}. However, we note that while this application is well-used, to our knowledge, there is no existing literature that details a practical algorithm for robust extrapolation and its implementation.

To bridge this gap, we propose a practical method for constructing a suitable prolongation operator for \VPD, using barycentric extrapolation for fine-mesh boundary vertices that lie outside coarse elements. Concretely, our goal is to associate every fine-mesh vertex with a corresponding coarse element that allows the computation of barycentric coordinates, similar to the conforming case. For fine vertices already contained within a coarse element, we simply assign their current coarse element for binding. For those outside the coarse mesh, we identify a suitable nearby coarse element to bind to, and apply negative barycentric coordinates for extrapolation. Once these bindings are established, barycentric coordinates can be computed for all these vertices, and the prolongation matrix, although it will contain negative entries, functions identically to the conforming case discussed earlier. 

\paragraph{Filtering bindings using geodesic distance} However, a major challenge with na\"ive binding methods, such as closest-point binding based on Euclidean distance \cite{Georgii:2006:MG4RTDefo}, is their lack of awareness of the mesh topology. This, in turn, can lead to fine vertices being mistakenly bound to nearby but geodesically distant coarse elements. This often results in topologically incorrect bindings and introduces dramatic animation artifacts. To mitigate this issue, we present a novel robust binding algorithm that avoids these misbindings (see Figure \ref{fig:binding-viz}). The core idea is that neighboring fine vertices should not be bound to geodesically distant coarse elements, as this would be clearly counterintuitive. Instead, our algorithm uses neighborhood information to infer coarse element bindings while incorporating topological awareness.

\subsubsection{Prolongation Construction (Barycentric Case)} 
Consider a simple two-level hierarchy consisting of a coarse mesh \(\mathcal{M}_C\) with vertices \(\mathcal{V}_C\)  and tetrahedra \(\mathcal{T}_C\), and a fine mesh \(\mathcal{M}_F\) with vertices \(\mathcal{V}_F\) and tetrahedra \(\mathcal{T}_F\), representing the same shape. The algorithm assigns fine mesh vertices \(\mathcal{V}_F\) to their corresponding coarse elements in \(\mathcal{T}_C\), which are then used to compute barycentric coordinates. First, it precomputes the vertex-vertex adjacency for the fine mesh \(\mathcal{M}_F\) and the element-element adjacency for the coarse mesh \(\mathcal{M}_C\), to build the necessary connectivity graphs. For each fine vertex, the algorithm checks whether it is already within a coarse element. If so, the vertex is assigned directly to that element. This binding is treated as the ground truth, i.e., the correct correspondence. No misbinding occurs if a fine vertex already lies inside a coarse element, since enforcing the non-misbinding condition (which ensures geodesic consistency) is typically manageable during hierarchy construction (e.g., decimation or tetrahedralization with proper constraints). In practice, we observed no such issues across our examples. Any fine vertex that falls outside the coarse mesh is collected and marked as unassigned for further processing.

\setlength{\columnsep}{0.5em}
\setlength{\intextsep}{0em}
\begin{wrapfigure}{r}{69pt}
   \includegraphics[width=69pt]{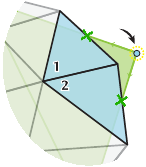}
\end{wrapfigure}
The algorithm then iteratively processes unassigned vertices. It first selects a vertex (among possibly many) with the largest number of neighbors already assigned to coarse elements. For each selected vertex, ray-mesh intersection tests are performed against the coarse mesh surface to identify first-hit intersections. Rays are cast along the incident edges of the fine vertex toward the opposite endpoints (see inset). If an intersection exists, the closest intersected triangle is assigned as the vertex's home coarse element. Otherwise, the algorithm gathers home triangles from the vertex's neighbors and assigns the closest one as its home coarse element. Once assigned, the update propagates to neighboring vertices in the fine mesh adjacency graph. The process repeats until all fine vertices are assigned, resulting in a complete mapping of \(\mathcal{V}_F\) to \(\mathcal{T}_C\). 

It is useful to note that, although we have described the two meshes as fine and coarse and built the binding from the fine vertices to the coarse elements, their roles are interchangeable. The same binding process can be applied in the reverse direction (coarse elements to fine vertices) using the same method, which is, for instance, useful when binding with Biharmonic Coordinates. Finally, while this method is described for a two-level hierarchy $(\mathcal{V}_C, \mathcal{T}_C)$ and $(\mathcal{V}_F, \mathcal{T}_F)$, it can be trivially extended to an $L$-level hierarchy by applying the same procedure iteratively between adjacent levels.

\subsubsection{Prolongation-related Error Analysis} Since our method uses negative barycentric weights to construct the prolongation matrix, it is important to show that the effects of this extrapolation remain bounded, that is, the entries of the prolongation matrix $P$ exert only a limited influence on the finer states. In particular, we must ensure that the effect of negative weights itself does not accumulate over time or lead to unbounded growth in \VPD. To support this, we provide an analysis in Appendix \ref{app:errorAnalysisForFunAndProfit}, demonstrating that our use of negative barycentric weights remains safe and well behaved, in addition to our animation results.

\begin{figure}[t!]
  \centering
  \includegraphics[width=\linewidth,keepaspectratio]{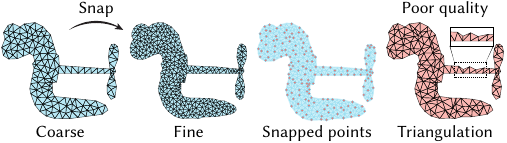} 
  \caption{{\bf Snapping coarse boundary vertices to fine-scale geometry} are practical workarounds for Biharmonic Coordinates but can degrade the triangulation quality of the original coarse mesh. Instead, we generalize Biharmonic Coordinates to accommodate control points that are not constrained to fine mesh vertices.}
  \label{fig:snapping-points}
\end{figure}

\subsection{Beyond Barycentric Prolongation}
\label{sec:alternative-interpolants}
The Progressive Dynamics framework is general enough to support a variety of interpolants for prolongation. We show that several commonly used interpolants in embedded simulation, each of which can benefit from the aforementioned extrapolation strategy, can be integrated into our prolongation framework as ``plug-and-play'' components, with only minor modifications to their original formulations. 
In this section, we detail the construction of prolongation operators using two alternate interpolation schemes: Biharmonic Coordinates \cite{wang2015linear} and Phong Deformation \cite{James2021} as examples. These alternative embedding methods do not improve finest-level \VPD results over linear barycentric coordinates. Rather, we demonstrate that they integrate naturally into our \VPD framework to yield \emph{higher-quality previews}.

\subsubsection{Prolongation using Biharmonic Coordinates}
\label{sec:prolongBiharmonic}
Skinning-based methods that utilize generalized barycentric coordinates have been widely used to smoothly deform embedded geometry by treating coarse shapes as control rigs \cite{Jacobson2011, joshi2007harmonic, Ju2005, floater2003mean, wang2015linear}. While effective for applications such as character articulation and cage-based deformation, these techniques are typically designed for specific control geometries (``cages'') and are not easily adaptable to arbitrary coarse meshes. Taking Biharmonic Coordinates as an example, the original formulation allows a set of coarse vertices, where the simulation is executed, to act as control points. These coarse vertices are then used to interpolate a high-resolution output via Biharmonic Coordinates for smooth deformation.

However, a key limitation of the original Biharmonic Coordinates \cite{wang2015linear} that prevents its direct integration into our framework is the assumption that all coarse control vertices must coincide with fine mesh vertices. This constraint is difficult to satisfy within our tetrahedral mesh hierarchy, where enforcing such alignment during construction is impractical. A common workaround, snapping the coarse vertices to the nearest fine ones, often leads to deteriorated mesh quality, including skinny and degenerate triangles, especially when the coarse and fine mesh resolutions are similar (see Figure~\ref{fig:snapping-points}). Therefore, it is necessary to modify Biharmonic Coordinates to function effectively with our mesh hierarchy.

We observe that Biharmonic Coordinates can be extended to support non-coincident control points by modifying the binary selector matrix $S$ (as described and defined in Equation 2 in \cite{wang2015linear}) to incorporate barycentric interpolation. Furthermore, for vertices that lie outside the fine tetrahedral mesh, we can further apply the extrapolation strategy discussed previously. This adaptation allows us to reformulate the computation of Biharmonic Coordinates $W$ with non-coincident control points as follows:
\begin{align}
\min_{W} \ \mathrm{trace}\Big(\frac{1}{2} W^\top A W\Big) \quad \mathrm{subject \ to} \quad BW = I,
\end{align}
where $A$ is a squared Laplacian smoothness energy (a positive semi-definite matrix) and $B$ is the linear interpolation matrix constructed using barycentric coordinates such that $B \, V_f = V_c$, which may include negative weights for boundary extrapolation. This reformulation enables the use of Biharmonic Coordinates in the context of our tetrahedral mesh hierarchy while addressing the challenges posed by non-coincident control points and non-conforming boundaries. We show that $V_f = W \, V_c$, which enforces interpolation constraints at selected control vertices (coarse vertices, in this case) and is equivalently satisfied by solving the above constrained quadratic programming problem; we refer the reader to Proposition \ref{prop} in Appendix~\ref{app:biharmonicProof}.

Finally, we show that our modified Biharmonic Coordinates maintain useful sparsity and locality in practice (see Figure~\ref{fig:locality-of-modified-bc}).

\begin{figure}[t!]
  \centering
  \includegraphics[width=\linewidth,keepaspectratio]{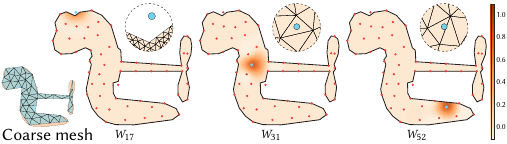} 
  \caption{{\bf Sparsity and locality of modified Biharmonic Coordinates:} We provide a visualization showing that our modification maintains the sparsity and locality of Biharmonic Coordinates with non-incident control points.}
  \label{fig:locality-of-modified-bc}
\end{figure}

\subsubsection{Prolongation using Phong Deformation} 
\label{sec:prolongPhong}
Another alternative interpolant for prolongation that we could consider is Phong Deformation \shortcite{James2021}, a simple, robust, and practical vertex-based quadratic interpolation scheme that offers improved visual quality over linear interpolation. Although it remains only \( C^0 \)-continuous as a linear interpolation, it substantially reduces visual artifacts in embedded geometry. The core idea is to first average element-based linear deformation models to the vertices, then barycentrically interpolate these vertex models while blending with the traditional linear interpolation. Meanwhile, Phong Deformation does not assume a strictly embedded hierarchy, which means that fine boundary elements are not necessarily enclosed by coarse elements. This makes our proposed approach a natural fit: we can construct the hierarchy using our method as described in Section~\ref{sec:vol-hierarchy} and apply extrapolation to handle boundary vertices located outside the fine mesh. Our framework then computes the blended vertex deformation gradient and element deformation based on the corresponding bound elements, enabling seamless integration of Phong Deformation into our pipeline.

Moreover, although Phong Deformation itself offers a fast, robust, and easy-to-implement alternative to linear interpolation—with notably improved visual results for irregular tetrahedral meshes—it is ultimately a simple extension of linear interpolation for embedded mesh deformation. As such, it lacks contact awareness and higher-order accuracy. Thus, our \VPD framework—designed to be compatible with the Phong Deformation as the interpolant—can also be viewed as a physical relaxation tool for postprocessing Phong Deformation results, correcting embedding-related artifacts such as self-intersections and recovering additional physical effects such as high-frequency motion.

\section{Evaluation}
\label{sec:evaluation}
We use Apple's Accelerate solver for linear solves and Eigen for the remaining linear algebra routines \cite{eigenweb}. We report example statistics and timings on an Apple MacBook Pro with an M4 Max chip and 128 GB of RAM.

\subsection{Comparisons}

\begin{figure}[t!]
  \centering
  \includegraphics[width=\linewidth,keepaspectratio]{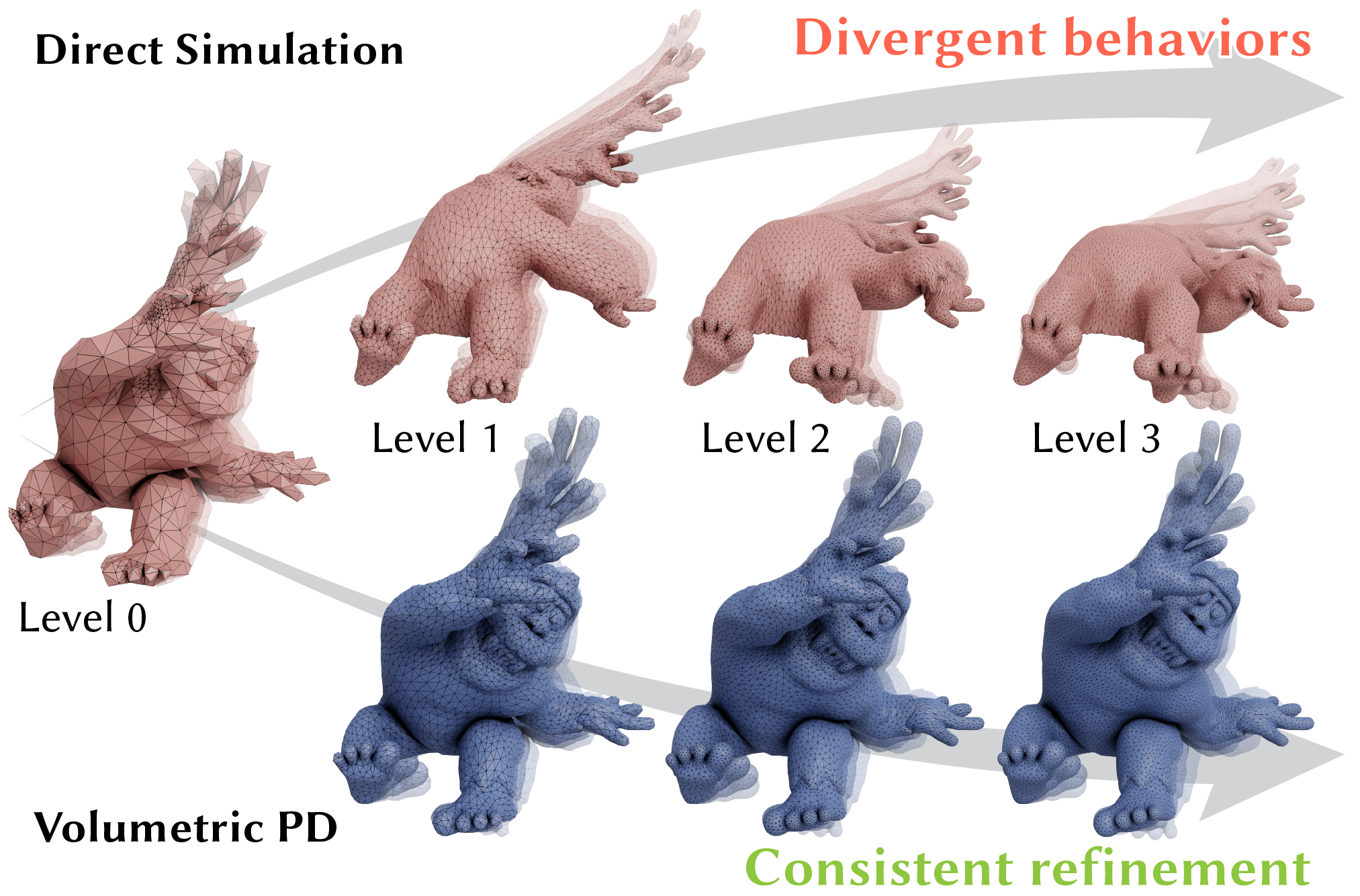} 
  \caption{{\bf Comparison with direct simulation:} A constrained character is subject to a strong impulsive loading that causes a rapid motion. (Top) Direct simulation responses quickly diverge across resolution levels, whereas (Bottom) \VPD produces similar responses across levels and closely matches the coarse level-0 preview simulation. }
  \label{fig:ape-divergence}
\end{figure}

\paragraph{Direct Simulation.} For some of our benchmark examples (see Figures~\ref{fig:beachball-direct} and \ref{fig:ape-divergence}), we perform the corresponding direct simulations at each resolution level using the same underlying IPC simulator~\cite{Li2020IPC} integrated into our testing framework. As in prior comparisons between shell-based Progressive Dynamics~\cite{zhang2024progressive} and direct C-IPC simulations~\cite{Li2021CIPC}, we observe a significant divergence across resolutions when using direct simulation. Figure~\ref{fig:ape-divergence}, together with our supplemental video, highlights substantial discrepancies in shape and trajectory between different resolution simulations, which emerge after only a short period of simulation time. Despite identical simulation setups—including the same material parameters and initial conditions—these IPC results exhibit markedly different material behavior and motion. This sensitivity to resolution is a well-known limitation of direct simulation methods and is not unique to high-fidelity finite element schemes \cite{zhang2024progressive, zhang2025progressive}. In contrast, \VPD yields qualitatively consistent animations across all levels of detail, with progressively enhanced visual fidelity at higher resolutions. Our figures and video results demonstrate frame-by-frame evidence of consistent behavior in both fine geometric details and overall motion.

\paragraph{TRACKS} To further evaluate the benefits of \VPD in terms of producing rich and consistent LOD enrichment, we compare it against a TRACKS-style method adapted from the original formulation for shell simulation~\cite{bergou2007tracks}. Although TRACKS was originally developed for shell animations, its core idea, enforcing position-based constraints to track a target animation, naturally extends to volumetric settings. In our experiment (Figure~\ref{fig:tracks-comparison}), we use the same slit-array object as in Figure~\ref{fig:increasing-weights} and construct a two-level hierarchy using $h = 0.01$s, where the fine-resolution simulation is constrained to follow the prolonged coarse result via a per-vertex penalty: $w \|x - P_{l}^{l-1} x^{t+1}_{l-1}\|^2_{M_l}$, with $P_{l}^{l-1}$ denoting the prolongation matrix. We find that the effectiveness of this TRACKS-style method is highly sensitive to the choice of the penalty weight $w$. When $w$ is not carefully tuned, the fine-level simulation can diverge from the coarse solution, leading to ghost force artifacts and degraded motion quality. In the same two-level setup, unlike the TRACKS-style approach, \VPD does not rely on penalty terms to enforce cross-resolution coherence. It naturally produces consistent animation at the fine level, with increasing geometric fidelity and stable motion as the resolution improves. See Figure~\ref{fig:tracks-comparison} and our supplemental videos for details. 


\begin{figure}[t!]
  \centering
  \includegraphics[width=\linewidth,keepaspectratio]{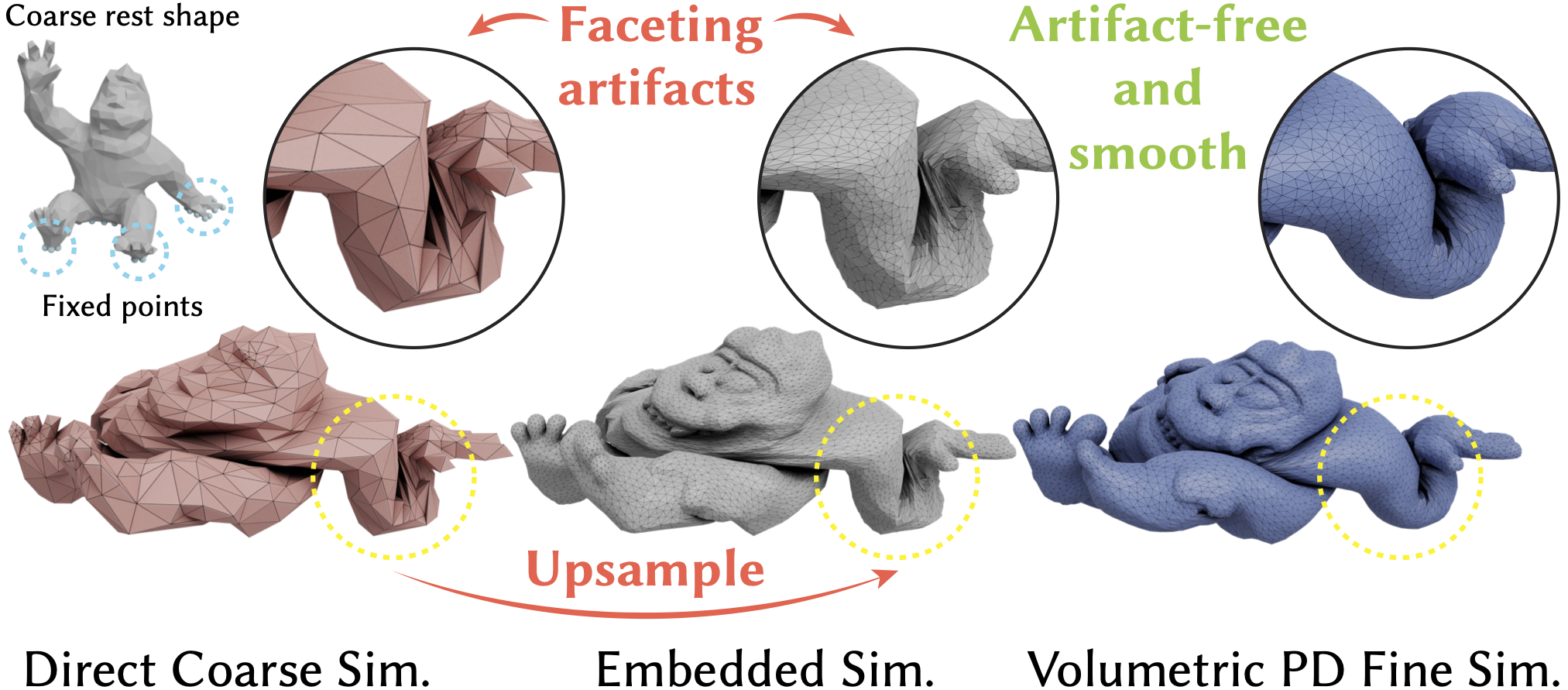} 
  \caption{{\bf Comparison with embedded simulation:} (Left) A coarse direct simulation used for (Middle) embedded deformation of a fine render mesh introduces visible ``faceting'' artifacts. (Right) In contrast, progressive refinement of the coarse simulation using \VPD produces high-quality fine-scale deformations (see text for discussion).
  }
  \label{fig:compare-to-embedded}
\end{figure}

\paragraph{Embedded Simulation.} As shown in Figure~\ref{fig:compare-to-embedded}, running a coarse simulation followed by direct embedding often leads to faceting artifacts, where large, planar deformations persist even at fine resolutions. In tighter contact scenarios (Figures~\ref{fig:dice-and-wizard} and \ref{fig:bc-comparison}), embedded simulations—broadly defined to include approaches such as applying skinning methods to upsample direct simulation results—can also lead to gaps and self-intersections, as the fine mesh may not align with the coarse geometry, particularly on curved surfaces. These issues reflect a broader limitation of commonly used interpolants in embedded simulation, whether linear, higher-order, or skinning-based, which lack physical awareness and fail to capture fine-scale motion or resolve contact accurately. Our \VPD framework addresses this by treating prolongation as a modular component: these interpolants can therefore be used as ``plug-and-play'' operators within a multilevel simulation, with \VPD serving as a physical relaxation stage that enriches embedded results with detailed dynamics and improved contact handling.

\subsection{Benchmark Examples}
To demonstrate the effectiveness of \VPD, we evaluated it in a wide range of benchmark examples with varying levels of complexity. Unless otherwise specified, we use barycentric coordinates with our boundary extrapolation as the default prolongation method. Our results show that the progressive volumetric simulation framework produces high-fidelity, LOD-consistent animations across resolutions, even under challenging dynamics such as high speeds, large deformations, and frictional contact. We refer to our supplemental video for the full set of animations.

\paragraph{Large Deformation.} To demonstrate that \VPD can robustly handle large extreme deformations over long time spans (up to 800 timesteps for 8 seconds), we begin with a simple setup: a single gorilla character simulated using a four-level tetrahedral hierarchy with timestep size $h = 0.01$ and relatively soft material ($Y = 2 \times 10^4$ Pa). We also enforce fixed Dirichlet conditions on his bottom, one hand, and both feet. At the initial frame, we apply a strong upward impulse that stretches the character to an extreme configuration, inducing rapid oscillatory motion. As shown in Figure~\ref{fig:ape-divergence} and our supplemental video, \VPD maintains stable and consistent behavior across all levels, exhibiting strong frame-by-frame agreement and overall trajectories even under severe deformation.

\begin{figure}[t!]
  \centerline{\includegraphics[width=\hsize]{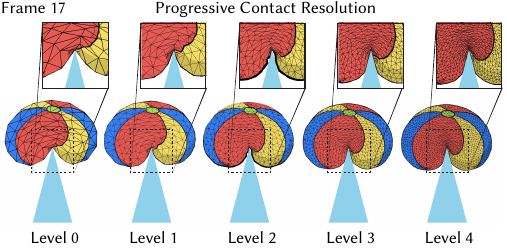}} 
  \caption{{\bf Ball on sharp spikes:} We show that \VPD effectively resolves contact with sharp spikes, with higher resolutions producing boundary surfaces that more precisely conform to the spike geometry.}
  \label{fig:beachball-spike}
\end{figure}

\paragraph{Sharp Spike Contact.} As shown in Figure~\ref{fig:beachball-spike}, we examine a 2D scenario in which a soft beach ball is dropped onto a sharp spike-shaped collider, using a five-level hierarchy with vertex counts of 0.06K, 0.13K, 0.31K, 0.57K and 1.1K and a timestep of $h = 0.01$s. At coarser levels, the limited degrees of freedom do not closely conform to capture the details of the spike's indentation, resulting in an oversimplified deformation. In contrast, finer resolutions produce boundary geometry that more accurately conforms to the collider shape, while preserving the time of maximum compression. This example illustrates how \VPD supports progressively refined contact interactions, enabling an increasingly realistic and detailed volumetric response as resolution improves.

\begin{figure}[t!]
  \centering
  \includegraphics[width=\linewidth,keepaspectratio]{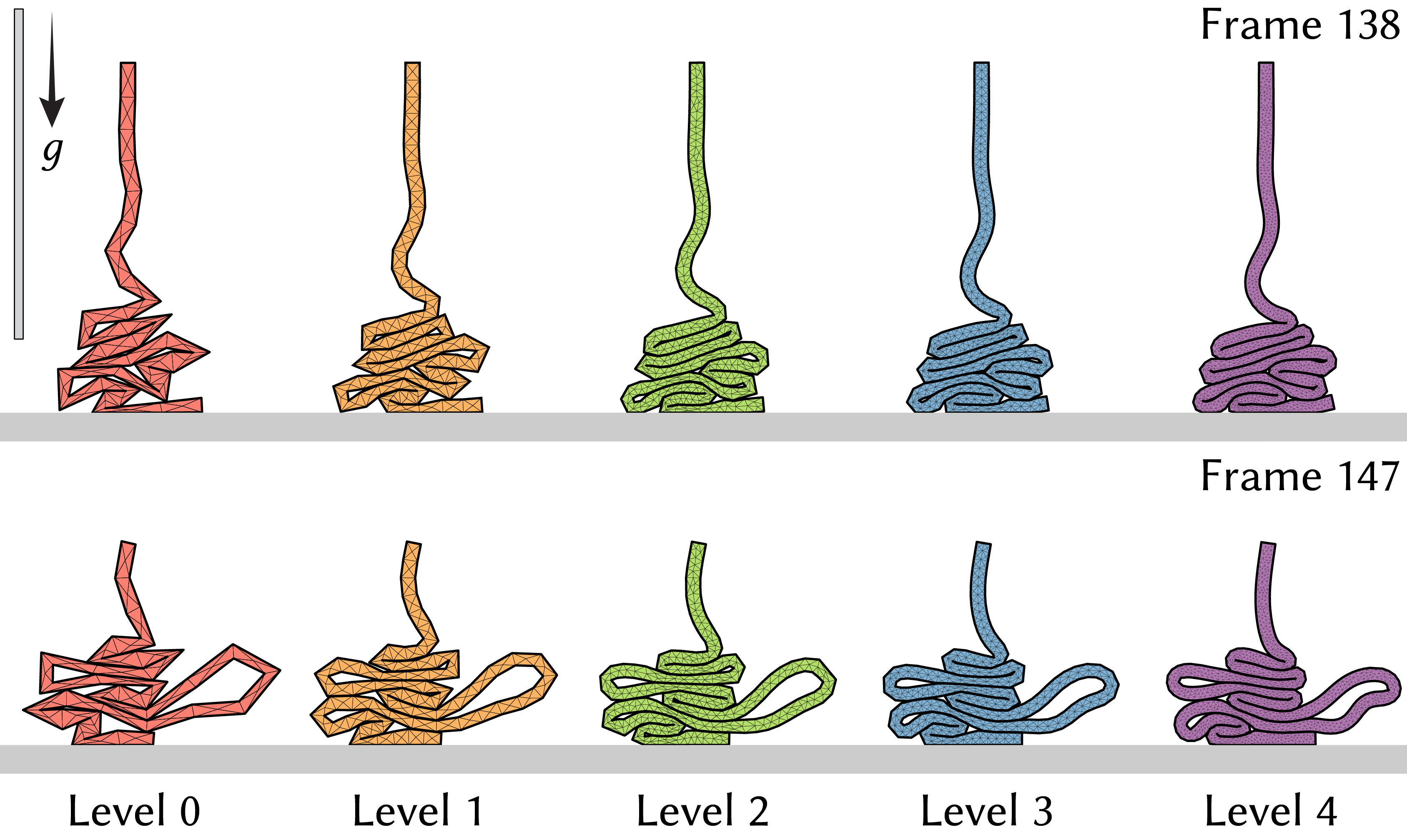} 
  \caption{{\bf Rope drop:} A rope is dropped under gravity, colliding with the ground. At the coarse level, the limited degrees of freedom lead to noticeable kinks, which are gradually refined and resolved into a smooth shape at the finest level as the \VPD progress.}
  \label{fig:rope-drop}
\end{figure}

\paragraph{Rope Drop.}
As a general purpose framework for volumetric animation, \VPD is effective on a wide range of geometries, including both volumetrically dense objects and thin structures, e.g., with high aspect ratios. To illustrate this, we simulate a rope-like object falling under gravity, colliding with the ground, and then tightly coiling, using a five-level hierarchy with vertex counts of 0.1K, 0.2K, 0.4K, 0.8K and 1.6K and a timestep of $h = 0.01$s. As the rope settles, the coarse levels—limited by low degrees of freedom—exhibit noticeable kinks when folding. As the simulation progresses through finer levels, \VPD incrementally refines the geometry, resolving these artifacts into smooth, physically plausible deformations. The final result, shown in Figure~\ref{fig:rope-drop}, shows a well-formed natural pile at the highest resolution. This example highlights the robustness of \VPD{} in capturing realistic dynamics, even in slender volumetric objects, without special treatment or specific tuning for thin structures.

\paragraph{Extreme Coarse Shapes.}
To evaluate \VPD's support of extremely coarse inputs for fast previewing, we present two 2D examples in Figure~\ref{fig:extreme-coarse-shapes}, each built with a four-level hierarchy. The coarsest levels are aggressively decimated to just 0.07K and 0.06K vertices—over 50$\times$ fewer than their respective finest-level counterparts. As a result of this simplification, the coarse shapes differ significantly from the fine-level geometry, exhibiting nonconforming boundaries that capture only rough polyhedral approximations of the fine-level shape's contour. Despite this large discrepancy, \VPD successfully introduces additional physical and geometric details at each level while preserving overall coherence and consistency of motion throughout the animations. This example highlights the robustness of our approach, demonstrating that \VPD can utilize highly simplified coarse inputs without requiring strict geometric alignment across the hierarchy.

\begin{figure*}[t!]
  \centering
  \includegraphics[width=\linewidth,keepaspectratio]{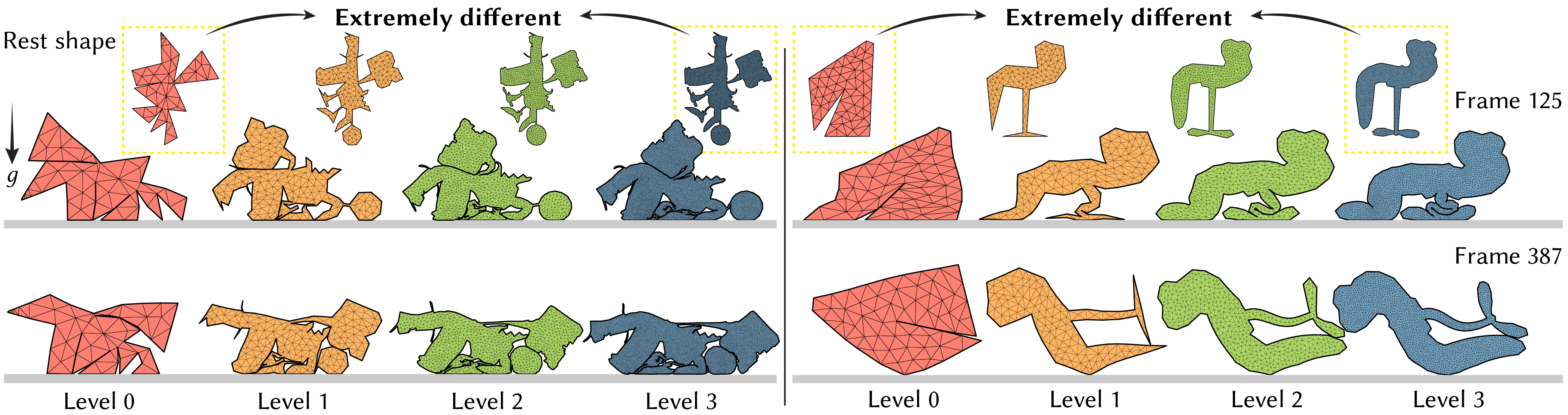} 
  \caption{{\bf Progress even with extreme coarse shapes:} We demonstrate that even in cases involving extreme coarse shapes with significant discrepancies between representations at different levels, \VPD effectively introduce additional physical and geometric details at each step while preserving reasonable geometric consistency.}
  \label{fig:extreme-coarse-shapes}
\end{figure*}

\begin{figure}[t!]
  \centering
  \includegraphics[width=\linewidth,keepaspectratio]{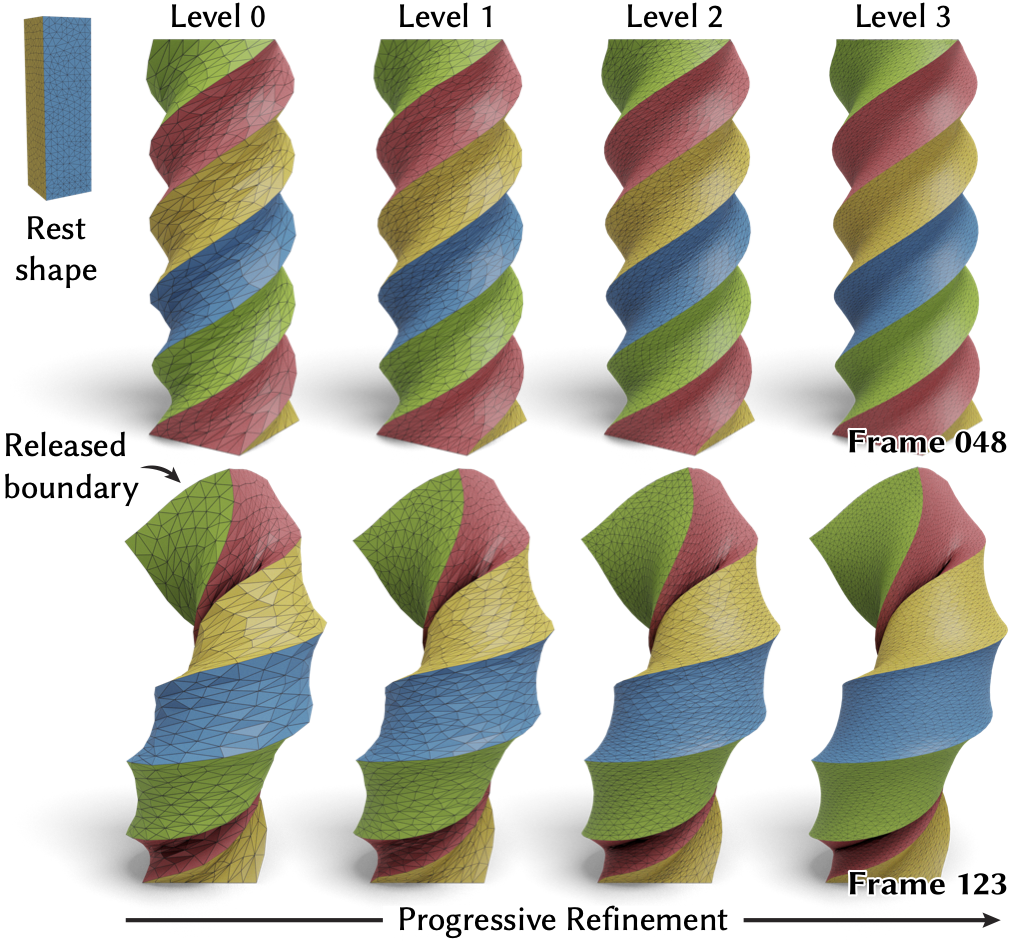} 
  \caption{{\bf Twist test:} (Top) A steadily twisted bar shows good correspondence across resolutions when extremely deformed, but also when (Bottom) it is released and undergoes a rapid untwisting motion. }
  \label{fig:twisting}
\end{figure}

\paragraph{Twisting Behavior.}
As a standard stress test, we include a 3D example in which a bar-shaped object undergoes a large imposed twisting (see Figure~\ref{fig:twisting}). The animation uses a four-level hierarchy with vertex counts of 3.7K, 9.1K, 32K, and 107K, respectively per level, and a timestep size of $h=0.01$s. Twisting is induced by rotating both ends of the bar in opposite directions. Despite the severity of the deformation, \VPD maintains strong frame-by-frame consistency across all levels, while ensuring intersection-free geometries throughout. As resolution increases, the twisted edge geometries are more sharply resolved and the remaining surfaces smoother. When one end of the bar is released at a later frame, the resulting dynamic response exhibits consistent transition timing and overall behavior across all levels, further demonstrating \VPD{}’s robustness under extreme, time-dependent deformation.

\begin{figure}[t!]
  \centering
  \includegraphics[width=\linewidth,keepaspectratio]{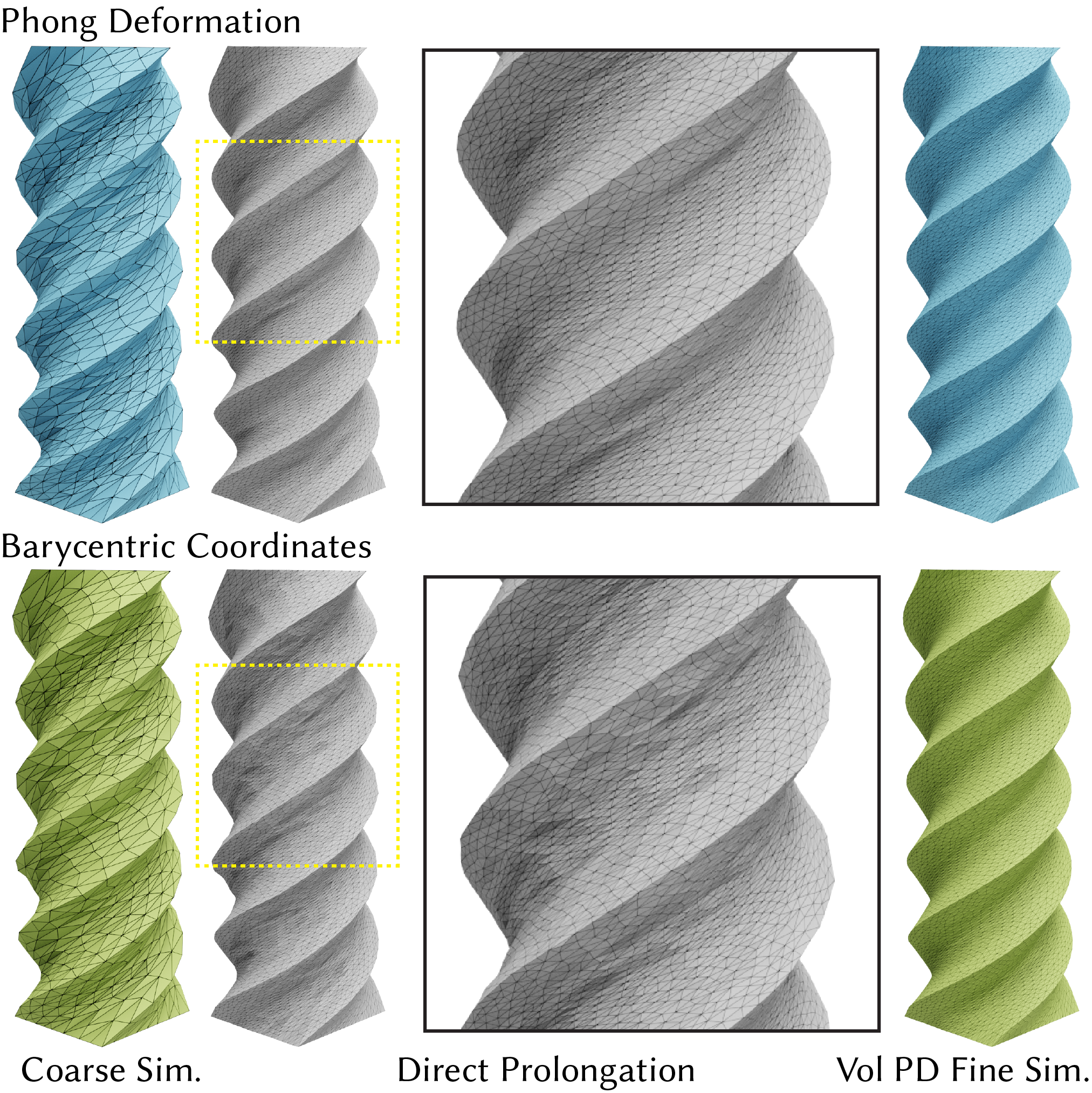} 
  \caption{{\bf Comparison with using Phong Deformation:} (Left) A coarse simulation mesh is directly prolongated using (Top,gray) Phong Deformation and (Bottom,gray) barycentric coordinates to obtain a fast preview, however the latter embedding has more obvious faceting artifacts. In contrast, the final fine-scale Vol PD simulations using Phong Deformation (Top,Right) and barycentric prolongation (Bottom,Right) both have visually comparable high-quality results.}
  \label{fig:phong-comparison}
\end{figure}

\paragraph{Extreme Geometry Differences Across Levels.} As discussed earlier, detailed and curved volumetric objects generally introduce larger geometric discrepancies across levels due to the lack of alignment between coarse and fine representations. Beyond the 2D example in Figure~\ref{fig:extreme-coarse-shapes}, we further demonstrate \VPD{}’s robustness in handling such large discrepancies through a challenging 3D test in Figure~\ref{fig:dice-and-wizard}. Here, a gnome character is dropped onto a letter block collider using a two-level hierarchy. To push the mismatch to an extreme, we force the geometry of the block to be just a simple cube at the coarsest level, while the finest level geometry includes deep, intricate letter engravings embedded into the block's surface that can and will change the simulation's collision dynamics. We test both a soft block collider ($Y = 2 \times 10^5$, shown in the first row) and a stiff one ($Y = 1 \times 10^{8}$, shown in the second row). In both cases, the gnome animation adapts to the geometry at each level, with the fine-level simulation conforming smoothly to the detailed surface, without locking, drifting, or violation of motion consistency. This example highlights \VPD{}’s strong resilience to geometric discrepancies while maintaining consistent and physically plausible behavior across resolutions.

\setlength{\columnsep}{0.5em}
\setlength{\intextsep}{0em}
\begin{wrapfigure}{r}{100pt}
   \includegraphics[width=100pt]{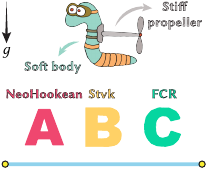}
\end{wrapfigure}
\paragraph{Heterogeneous Materials and Multi-Object Interaction.} Beyond single-object simulations with homogeneous materials, we also demonstrate \VPD{}’s capability in handling heterogeneous materials and complex multi-object interactions. As shown in Figure~\ref{fig:heterogeneous-material}, this 2D example features a soft-bodied worm character (Neo-Hookean, $Y = 2 \times 10^4$), a rigid propeller (effectively stiff with $Y = 2 \times 10^{10}$), three letter-shaped objects (A, B, and C) with increasing stiffness values ($Y = 4 \times 10^4$, $5 \times 10^5$, and $6 \times 10^6$) and different material models (Neo-Hookean, St. Venant–Kirchhoff, and co-rotational elasticity), and a bouncy trampoline. Material consistency across the hierarchy is maintained by first assigning materials to the elements at the coarsest level and then propagating those assignments to the corresponding elements at finer levels. We then define the elastic potential independently at each level based on these assignments. As these objects fall under gravity, collide with the trampoline, and bounce off, \VPD maintains consistent global behavior across resolutions. Deformations become increasingly detailed and realistic with refinement, highlighting \VPD{}’s robustness in resolving complex material behaviors and interactions within a unified multilevel framework. See the accompanying video for the animation.

\subsection{Prolongation Methods}

\begin{figure}[t!]
  \centering
  \includegraphics[width=\linewidth,keepaspectratio]{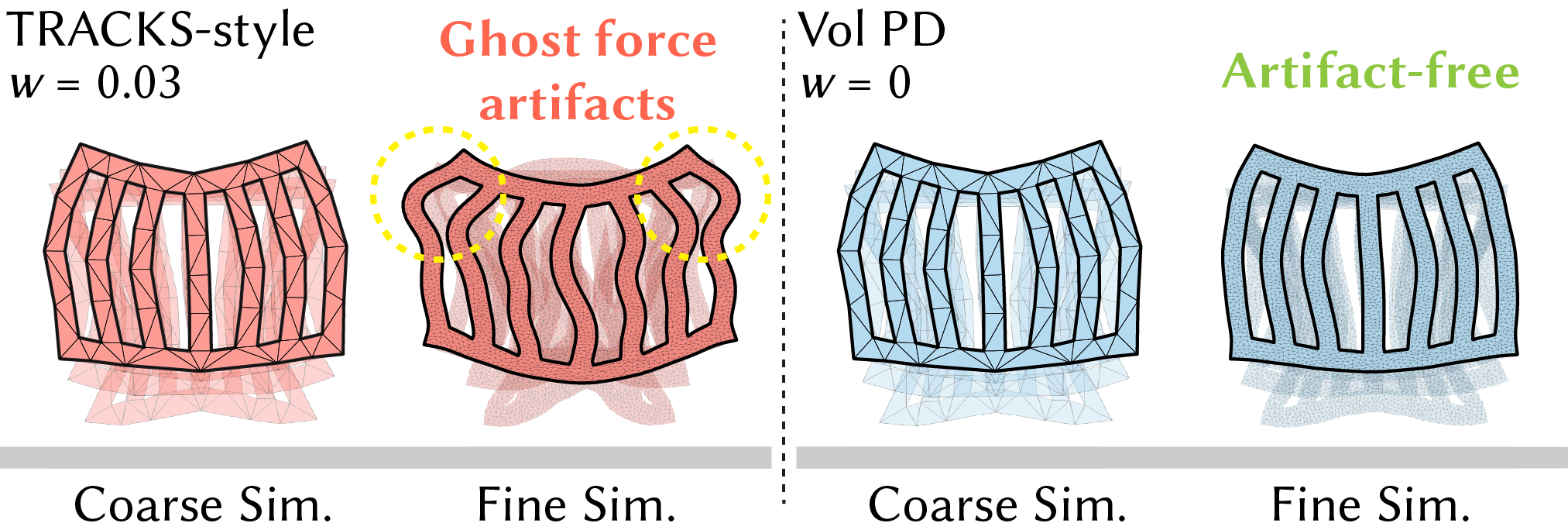} 
  \caption{\textbf{Comparison with TRACKS-style position-based tracking:} (Left) TRACKS-style control can produce undesirable deformable artifacts due to the coarse moment-based control forces, whereas (Right) \VPD estimates a fine simulation that closely resembles the coarse preview simulation. }
  \label{fig:tracks-comparison}
\end{figure}

\paragraph{Biharmonic Coordinates.} We compare \VPD{}'s results using Biharmonic Coordinates \cite{wang2015linear} against the default choice of barycentric coordinates. Using the same rope-drop setup as in Figure~\ref{fig:rope-drop} above, we observe that directly upsampling the coarse animation—particularly in regions of large deformation and severe stretching near corners—using Biharmonic Coordinates results in noticeable bulging artifacts that are not physically plausible. Moreover, because this direct upsampling is not contact-aware, it leads to severe intersection artifacts, as illustrated in Figure~\ref{fig:bc-comparison}. In contrast, incorporating Biharmonic Coordinates as a plug-in method for velocity prolongation during the construction of the inertia term is fully compatible with our framework and yields qualitatively similar refinement results to those obtained with barycentric coordinates. Additional quantitative comparisons are provided in Section \ref{sec:quant}.

\paragraph{Phong Deformation.} Similarly, Phong Deformation \cite{James2021} can likewise be used as an alternative plug-in method for velocity prolongation within \VPD and offers higher prolongation accuracy. While Phong Deformation similarly lacks physical awareness when used for direct upsampling, it produces high-quality results for finer levels when integrated with \VPD, achieving visual quality comparable to the default barycentric approach, as shown in Figure~\ref{fig:phong-comparison}. We refer to Section \ref{sec:quant} for a more detailed quantitative analysis.

\subsection{Quantitative Analysis}\label{sec:quant}
To quantitatively evaluate \VPD{}'s effectiveness and key properties, we adopt the temporal continuity and geometric consistency measures proposed in Zhang et al.~\shortcite{zhang2025progressive}, which assess per-level (``horizontal'') temporal continuity and across-level (``vertical'') geometric consistency over the mesh hierarchy. See Appendix~\ref{sec:metrics} for their detailed definitions.

\begin{figure}[t!]
  \centering
  \includegraphics[width=\linewidth,keepaspectratio]{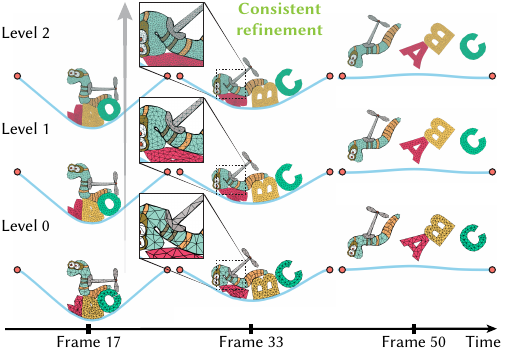} 
  \caption{{\bf Heterogeneous materials:} We show that progressive volumetric dynamics supports heterogeneous materials by simulating a scene that simultaneously involves Neo-Hookean, St. Venant–Kirchhoff, and co-rotational elastic materials of varying stiffnesses (see text for details).}
  \label{fig:heterogeneous-material}
\end{figure}

\paragraph{Continuity Measure.}
To quantitatively analyze and compare continuity preservation across methods, we use the temporal continuity metric proposed by Zhang et al.\shortcite{zhang2025progressive}. This measure offers a reliable means of assessing physical continuity and detecting temporal artifacts in arbitrary physical animations. We apply it to evaluate results generated using barycentric coordinates and Biharmonic Coordinates (see Figure~\ref{fig:bc-metrics}). It effectively demonstrates that using barycentric coordinates and Biharmonic Coordinates for prolongation within \VPD yields similar results in terms of continuity, neither outperforms the other, suggesting that both interpolants are equally effective and well suited for integration into the \VPD framework.

\paragraph{Consistency Measure.} Achieving consistent refinement across resolution levels is a core feature and goal of the Progressive Dynamics framework, and the same holds true in our volumetric extension. To evaluate geometric consistency across levels, we adopt the consistency metric proposed by Zhang et al.~~\shortcite{zhang2025progressive}. This metric essentially measures how well a coarser-level timestep captures the bulk deformation of its corresponding finer-level timestep, assuming that finer levels primarily contribute high-frequency details that do not alter coarse-scale behavior. Figure~~\ref{fig:bc-metrics} compares results generated by barycentric and Biharmonic Coordinates, evaluated using this metric.

Meanwhile, as noted by Zhang et al., while the \VO integration method used in \VPD generally yields stable, continuous, and consistent results, per-frame geometric consistency can still break under some extreme conditions, e.g., large time steps used with a deep hierarchy. To explore this limitation, we construct a targeted stress-test example and evaluate the influence of different consistency penalty weights $w$ (defined in Equation~\ref{eq:quadric-penalty}, with $w = 0, 0.025, 0.2, 0.4, 0.6$ and timestep $h = 0.04$) over an 8-level hierarchy. We apply the consistency metric to these results and discuss the impact of the penalty term in \S\ref{sec:conclusion} Limitations.

\begin{figure}[t!]
  \centering
  \includegraphics[width=\linewidth,keepaspectratio]{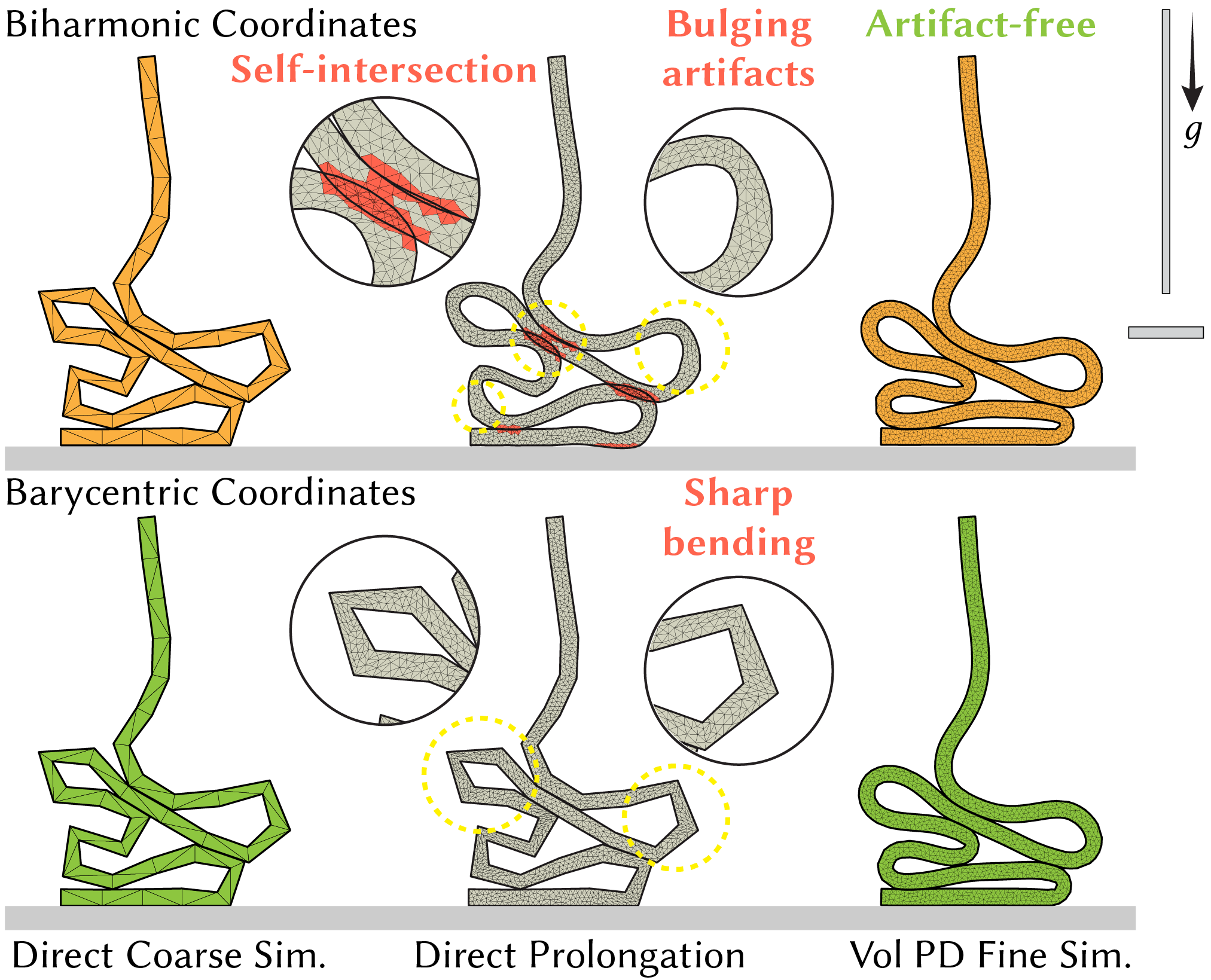} 
  \caption{{\bf Comparison with using Biharmonic Coordinates:} (Left) A coarse simulation of rope drop is only one element wide and in need of refinement. (Middle) Direct prolongation of the coarse sim to the fine scale either using (Bottom) barycentric interpolation retains contacts by exhibits poor smoothness, whereas (Top) Biharmonic Coordinates is smoother but introduces self-intersections and bulging artifacts. (Right) However, \VPD can use both prolongation methods to progressively simulation high-quality fine-scale results which are both smooth and intersection free.}
  \label{fig:bc-comparison}
\end{figure}

\begin{figure}[t!]
  \centering
  \includegraphics[width=\linewidth,keepaspectratio]{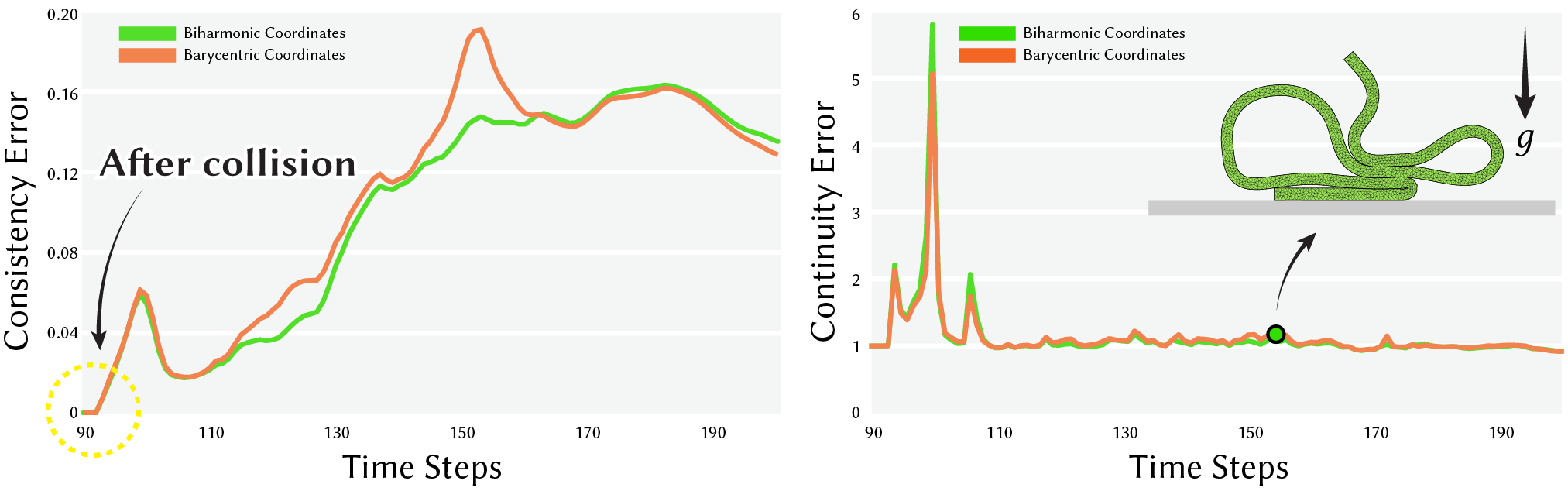} 
  \caption{{\bf Consistency and continuity metrics} are plotted for simulations using barycentric coordinates and Biharmonic Coordinates. In practice, both perform comparably well on this rope drop as well as other examples.}
  \label{fig:bc-metrics}
\end{figure}

\paragraph{Timings.} Since our only modification in \VPD{}’s per-level solve is the velocity integration step, each level's simulation runs at speeds comparable to direct simulation performed directly at that resolution. As for progressive dynamics for shells~\cite{zhang2024progressive,zhang2025progressive}, we emphasize that the analysis of speed-ups, and thus scalability, for Progressive Dynamics differs from that of standard simulation pipelines. In the context of Progressive Dynamics, scalability is realized through the ability of the framework to accelerate iterative design by allowing fast coarse-level solves. This allows animators to explore variations efficiently at low cost, reserving the expensive fine-level simulation for finalizing a selected design. This is particularly advantageous given that volumetric simulation is generally more computationally expensive than shell simulation because of the additional degrees of freedom in the interior. Across all examples in our paper, we consistently observe one to two orders of magnitude speedup when using our coarse preview-level simulation compared to direct fine-level simulation. In one extreme case—the teaser leaf sheep example—we achieve roughly a 120× speedup by running a single iteration of the coarse preview: the fine-level simulation takes nearly 20 hours, whereas the coarse-level version completes in just 10 minutes.

\subsection{Animation Design}

In addition to standard benchmarks, we also demonstrate an example of animation design tasks applied with \VPD, where fast coarse-level preview animations accelerate the iterative motion design process. We refer to our supplemental video for a detailed presentation of the animations throughout the design processes.

\paragraph{Bunny Noodles} We simulate a bowl of ``bunny noodles''—thin, volumetric deformable slices—dropping onto a plate under gravity, using a three-level hierarchy. As the animation progresses, the strands naturally form into a coherent pile, with fold details and surface regions becoming increasingly resolved at finer levels. The overall bulk motion remains stable and consistent throughout, allowing rapid iteration on coarse-level previews and producing high-quality results at the finest level suitable for final production. See Figure~\ref{fig:bunny-noodles} and our supplemental video for more details.

\begin{figure*}[t!]
  \centering
  \includegraphics[width=\linewidth,keepaspectratio]{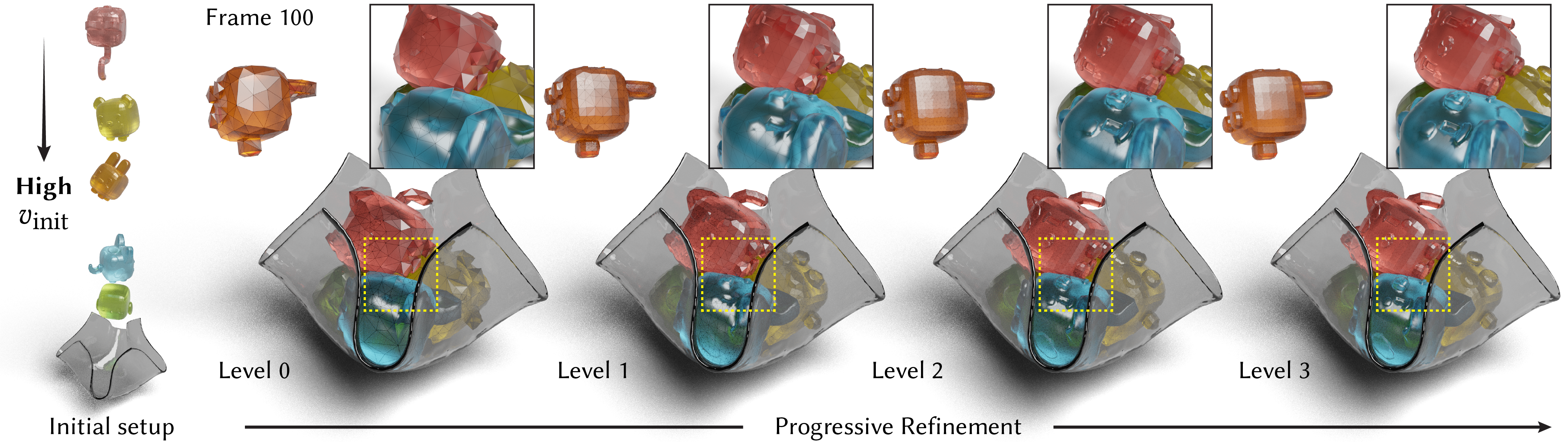} 
  \caption{{\bf Jello in the bowl:} High-speed collisions between Jello characters are consistently resolved across time and resolution. Please see the video for significant deformations. }
  \label{fig:jello}
\end{figure*}

\paragraph{Jello in the Bowl.} We simulate a set of soft, jello-like animal characters being thrown into a glass bowl at very high speed, using a four-level hierarchy. Due to the high initial velocity and low friction, the characters undergo substantial deformations and rapid rotations upon impact. At the coarsest level—100$\times$ lower in resolution than the finest—each character is represented by a simple, blobby approximation with only around 0.5K vertices and 1.5K tetrahedra. As \VPD refines through the hierarchy, both geometric (model surface details) and physical (deformation enrichment) features are progressively recovered, while maintaining consistent motion and structure across all levels. See Figure~\ref{fig:jello} and our supplemental video for more details.

\paragraph{Leaf Sheep.} We animate the story of two leaf sheep colliding with a mushroom-like plant and each other on a seafloor; here a strong underwater current causes them to collide and slide down, eventually settling together on a nearby rock. \VPD successfully refines the geometry and motion of the cerata (hair-like protrusions) with detailed motion, while maintaining closely matched trajectories across all levels of resolution. See Figure~\ref{fig:teaser} and our supplemental video for more details.
\section{Conclusion and Limitations}
\label{sec:conclusion}
We have extended the Progressive Dynamics framework \cite{zhang2024progressive,zhang2025progressive}
to volumetric finite elements, enabling efficient LOD animation design with predictive coarse previews. To support this, we introduce a practical method for multiresolution hierarchy construction and a simple, topology-aware algorithm for prolongation based on boundary binding, allowing several off-the-shelf interpolants to serve as plug-and-play components within our \VPD framework. Extensive stress tests—including high speeds, large deformations, and frictional contact—demonstrate the effectiveness and versatility of our proposed framework.

\paragraph{Limitations and Future Work} Despite these advances, there remain several limitations and opportunities to further improve both the efficiency and quality of progressive volumetric simulation.

\setlength{\columnsep}{0.5em}
\setlength{\intextsep}{0em}
\begin{wrapfigure}{r}{70pt}
   \includegraphics[width=70pt]{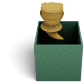}
\end{wrapfigure}
For example, as shown in Figure~\ref{fig:dice-and-wizard}, although \VPD generally handles geometric discrepancies across levels well, we can still construct extreme cases—such as using a hollow cube as the fine mesh—where the resulting fine animation remains consistent with the coarse preview but fails to produce physically plausible behavior, e.g., by falling into the cube (see inset). Meanwhile, penalty weights may still be necessary in extreme cases (see Figures~\ref{fig:increasing-weights} and \ref{fig:consistency-error}). We have been able to use direct energy evaluation instead of subspace proxy energies, but future work should investigate fast schemes for proxy energy evaluation such as adaptive quadrature, as well as homogenization techniques to mitigate artifacts such as shear locking with linear finite elements. Most of our implementation is CPU based, however, most parts of the algorithm would benefit from GPU acceleration. We have simulated volumes, but progressive simulation strategies for combining multiple co-dimensional entities are needed. Finally, interactive coarse-preview performance with consistent fine-scale refinements would be exceptionally useful for animation design and remains an open challenge.

\begin{figure*}[t!]
  \centering
  \includegraphics[width=\linewidth,keepaspectratio]{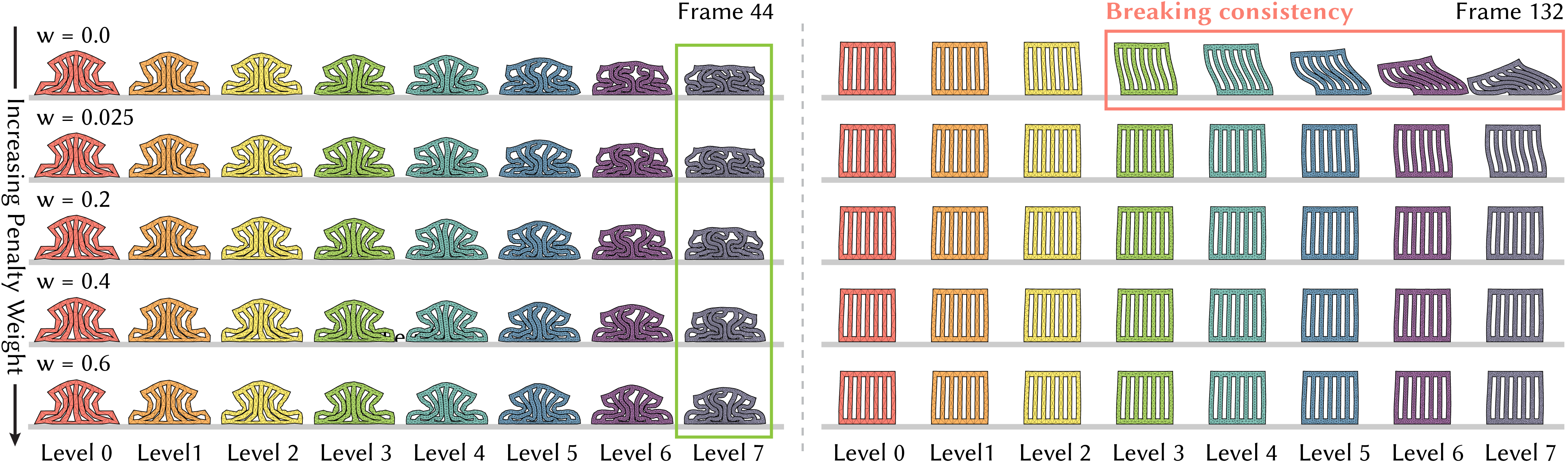} 
  \caption{\textbf{Shapes with Increasing Consistency Weights.} Using the same setup as in Figure~\ref{fig:consistency-error}, we visualize animation results with consistency penalty weights ($w = 0, 0.025, 0.2, 0.4, 0.6$, with $h = 0.04$) on an 8-level hierarchy of a deformable vertical slit-array object dropped onto the ground, shown at frames 44 and 142. As the consistency weight increases, the levelwise solutions become visibly more consistent. In contrast, without a penalty term ($w = 0$), consistency may break down as the animation proceeds. Even a small penalty ($w = 0.025$) significantly improves the preservation of consistency across levels.}
  \label{fig:increasing-weights}
\end{figure*}

\begin{figure}[t!]
  \centering
  \includegraphics[width=\linewidth,keepaspectratio]{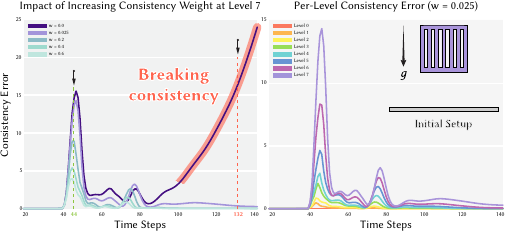} 
  \caption{\textbf{Measuring Consistency with Increasing Consistency Weights.} We evaluate consistency penalty weights ($w = 0, 0.025, 0.2, 0.4, 0.6$ with $h = 0.04$) on an 8-level hierarchy using a deformable vertical slit-array object dropped onto the ground, designed to stress-test per-frame consistency across resolutions. The left plot shows the impact of increasing consistency weight at the finest level ($\ell = 7$), where no penalty ($w = 0$) leads to visible consistency breaking as the animation progresses. Even a small penalty ($w = 0.025$) significantly improves consistency. The right plot reports the per-level consistency error across all levels using $w = 0.025$.}
  \label{fig:consistency-error}
\end{figure}

\appendix
\section{Biharmonic Coordinates with Noncoincident Control Points}
\label{app:biharmonicProof}

We reformulate the computation of biharmonic coordinates with noncoincident control points as follows:
\begin{align}
\min_{W} \ \mathrm{trace}\Big(\frac{1}{2} W^\top A W\Big) \ \mathrm{subject \ to} \ V_f = W V_c \quad \mathrm{and} \quad BV_f = V_c,
\end{align}
where $V_f = W V_c$ means interpolation constraints at selected control vertices (coarse vertices in our case) and $B$ is constructed using barycentric coordinates which potentially contain negative weights for extrapolation. We know that if $A$ has affine functions in its kernel, i.e., if $AV_f = 0$, then the weights $W$ will retain affine precision and we will naturally have $V_f = W V_c.$ Hence, it remains to be shown that the optimization problem above is equivalent to the following standard quadratic programming problem with linear equality constraints:
\begin{align}
\min_W \ \mathrm{trace}\Big(\frac{1}{2} W^\top A W\Big) \quad \mathrm{s.t.} \ BW = I,
\label{eq:SW}
\end{align}
where $I$ is the identity matrix.

Let $B$ be an arbitrary $m\times n$ matrix of rank $m$, and $A$ be a positive definite matrix. Let $I$ denote the $m\times m$ identity matrix.
Our goal is to solve \eqref{eq:SW} in this case.
Let $B^{-1}$ denote the right inverse of $B$. Then $BW=I$ is equivalent to $W=B^{-1}+T$ where $BT=\mathbf{0}$. Let $N$ be a matrix whose columns form an orthonormal basis of $\mathrm{Null}(B)$, i.e., $BN=\mathbf{0}$ and $N^\top N=I$. It follows that $T$ can be parameterized by $T=NY$. The problem \eqref{eq:SW} is then equivalent to 
\begin{align}
    \min_Y \mathrm{trace}\Big(\frac{1}{2}(B^{-1}+NY)^\top A(B^{-1}+NY)\Big).\label{eq:Y}
\end{align}
    Taking derivative of the objective in $Y$ shows that the minimizer to \eqref{eq:Y} is given by 
    $$Y=-(N^\top AN)^{-1}N^\top AB^{-1}.$$
The minimizer to \eqref{eq:SW} is thus given by 
$$W=B^{-1}-N(N^\top AN)^{-1}N^\top AB^{-1}.$$

An analogous argument works for the problem
\begin{align}
    \min_{V_f:BV_f=V_c}\mathrm{trace}\Big(\frac{1}{2}V_f^\top AV_f\Big),\label{eq:Vf}
\end{align}
simply by parametrizing $V_f=B^{-1}V_c+T$. The solution is given by
$$V_f=B^{-1}V_c-N(N^\top AN)^{-1}N^\top AB^{-1}V_c.$$
This leads to the following proposition.
\begin{proposition}\label{prop}
    In the above setting, denote by $V_f$ the solution to \eqref{eq:Vf} and by $W$ the solution to \eqref{eq:SW}. It holds that $V_f=WV_c$.
\end{proposition}

\section{Error Analysis of Negative Weights for Extrapolation}
\label{app:errorAnalysisForFunAndProfit}

The goal of this section is to show that the effect of negative weights in constructing the prolongation matrix does not accumulate over time. 
First, recall the construction of the momentum term using the velocity-only prolongation time integration scheme, defined as
\begin{align}\begin{split}
    \hat{x}_{l+1}^{t}
&= x_{l+1}^{t} + P_{l+1}^l v_{l}^{t} h= x_{l+1}^{t} + P_{l+1}^l (x_{l}^{t}  - x_{l}^{t-1}).
\end{split}\label{eq:x}
\end{align}
Meanwhile, considering that an implicit Euler timestep is resolved through minimization,
\begin{align}
x_l^{t+1} = \underset{x}{\operatorname{argmin}} \frac{1}{2 h^2} \|x-\hat{x}_l^t\|^2_{M_l} +  E_l(x),
\end{align}
we can concisely express $x_l^{t+1}$ as $\hat{x}_l^t + \gamma_l^{t+1}$, where $\gamma_l^{t+1}$ captures the nonlinear component of the minimization problem, explicitly,
\begin{align}
    x_l^{t+1} = \hat{x}_l^t + \gamma_l^{t+1}.\label{eq:x2}
\end{align}
Combining \eqref{eq:x} and \eqref{eq:x2} leads to
\begin{align}
    x_{l+1}^{t+1}
&= x_{l+1}^{t} + P_{l+1}^l (x_{l}^{t}  - x_{l}^{t-1}) + \gamma_{l+1}^{t+1}.\label{eq:x3}
\end{align}
Under this time integration scheme, we build the expression of $x_{l+1}^{t+1}$ recursively via the following lemmas. To facilitate our discussion, we assume that $t\geq l$. The case $t<l$ can be derived from a similar analysis.

\begin{lemma}
    It holds for all $t,l\geq 0$ that
    \begin{align}
        x^{t+1}_{l+1}=x_{l+1}^{1} + P_{l+1}^l (x_{l}^{t} - x_{l}^{0}) + {\sum^{t}_{i=1} \gamma_{l+1}^{t+2-i}}.\label{eq:1}
    \end{align}
\end{lemma}
\begin{proof}
 The result follows from fixing $l$ and inducting on $t$. We omit the details.
\end{proof}

\begin{lemma}\label{lemma}
    Suppose that $t\geq l\geq -1$. It holds 
    \begin{align}
    \begin{split}
        x_{l+1}^{t+1}&= P_{l+1}^{0} x_{0}^{t-l} + \sum^{l+1}_{j=1} P_{l+1}^{l+2-j} x_{l+2-j}^{1} - \sum^{l+1}_{j=1} P_{l+1}^{l+1-j} x_{l+1-j}^{0} \\
&\hspace{1cm} + \sum^{l}_{j=0} \sum^{t-j}_{i=1} P_{l+1}^{l+1-j} \gamma_{l+1-j}^{t+2-i-j}.
    \end{split}
        \label{eq:2}
    \end{align}
\end{lemma}

\begin{proof}
 The result follows from fixing $t-l$ and inducting on $l$ using \eqref{eq:1}. We omit the details.
\end{proof}


\sloppy It is important to note from Lemma \ref{lemma} that $x_{l+1}^{t+1}$ is essentially a weighted linear combination of the boundary states $x^{t-l}_0$, $x^{1}_{l+2-j}$, $x^{0}_{l+1-j}$, along with the nonlinear contributions introduced through optimization solves. By \eqref{eq:2}, the triangle inequality, and the sub-multiplicativity of the matrix 2-norm, we have
\begin{equation}
\begin{split}
\|x_{l+1}^{t+1}\|_2
&\le \|P_{l+1}^{0}\|_2 \ \|x_{0}^{t-l}\|_2 + \sum^{l+1}_{j=1} \|P_{l+1}^{l+2-j}\|_2 \ \|x_{l+2-j}^{1}\|_2 \\ 
&\hspace{1cm}+ \sum^{l+1}_{j=1} \|P_{l+1}^{l+1-j}\|_2 \ \|x_{l+1-j}^{0}\|_2  \\ 
&\hspace{1cm}+ \sum^{l}_{j=0} \sum^{t-j}_{i=1} \| P_{l+1}^{l+1-j}\|_2 \|\gamma_{l+1-j}^{t+2-i-j}\|_2.
\end{split}
\end{equation}
By far, this clearly indicates that the influence of the prolongation matrices on $x_{l+1}^{t+1}$ is inherently bounded by their norms. Taking one step further, note that $\|A\|_2 \leq \|A\|_F$ always holds, where the Frobenius norm is defined as $\|A\|_F = \sqrt{\sum_{i=1}^m \sum_{j=1}^n |a_{ij}|^2}$.

Based on this, we could further characterize the Frobenius norm of a prolongation matrix $P$ through its construction:

\paragraph{Case 1.} $P$ is constructed via barycentric coordinates where $\sum_j \lambda_{ij} = 1$ for all $i$ and $0 \leq |\lambda_{ij}| \leq 1$. Based on this, we know that $\|P\|_F = \sqrt{\sum_{i=1}^m \sum_{j=1}^n |\lambda_{ij}|^2} \leq \sqrt{m}.$

\paragraph{Case 2.} $P$ is constructed via barycentric coordinates where $\sum_j \lambda_{ij} = 1$ for all $i$, but $\lambda_{ij}$ can be negative. In this case, we know that some entries, $\lambda$, could be greater than 1, and thus their values become even larger when squared. 
\textit{However, since the decimation algorithm used to construct the mesh hierarchy includes a controllable parameter $\epsilon$, which defines the maximum allowable distance between a coarse vertex and the fine surface, we can ensure that these entries remain bounded.}

\paragraph{Case 3.} $P$ is constructed using other general methods, such as biharmonic coordinates, where the entries are typically allowed to be negative.

\section{Metrics}
\label{sec:metrics}

\subsection{Temporal Continuity Metric}
\label{sec:continuity-metric}

To evaluate temporal continuity of a proposed position $y_l^t$ in Progressive Dynamics, Zhang et al.~\shortcite{zhang2025progressive} define a continuity error relative to its time-stencil neighbors $y_l^{t-1}$ and $y_l^{t+1}$ using a midpoint state estimator derived from reformulating implicit Euler as a discrete boundary value problem. They then construct continuity error measures for each ``proposed'' position $y_l^t$ in the resolution-time grid relative to its horizontal (time) neighbors, $y_l^{t-1}$ and $y_l^{t+1}$, as
\begin{align}\label{eq:continuity-error}
\begin{split}
    e^t_l &= \|y^t_l - \phi^t_l(y^{t+1}_l,y^{t-1}_l)\|^2_{M_l}\\
&= \|\tfrac{1}{2} (y_l^{t+1} - 2 y_l^{t} + y_l^{t-1}) + \tfrac{h^2}{2} M_l^{-1}\nabla F_l(y_l^{t+1})\|_{M_l}^2,
\end{split}
\end{align}
which yields a per-timestep error in meters, integrated over the surface using the mass matrix $M_l$ to ensure resolution-aware scaling.

Although $e_l^t$ is well defined, in practice each timestep in Progressive Dynamics is solved to a fixed tolerance $\epsilon$ on the Newton decrement \cite{Li2021CIPC}, leading to varying residuals across steps. As a result, even direct single-level timestepping ends up with nonzero $e_l^t$. To account for this, they further normalize $e_l^t$ by the residual of the original progressive solve:
\begin{align*}
\hat{e}_l^t &= \|(x_l^{t+1} - \hat{x}_l^t) + h^2 M_l^{-1}\nabla F_l(x_l^{t+1})\|_{M_l}^2,
\end{align*}
where $\hat{x}_l^t$ corresponds to the update rule of the specific Progressive Dynamics integrator. Thus, they define the final continuity metric as:
\begin{align}
\label{eq:continuity-metric}
n_l^t = \frac{e_l^t}{\hat{e}_l^t},
\end{align}
a dimensionless measure for consistent comparison across timesteps and methods.

\subsection{Geometric Consistency Metric}
\label{sec:consistency-metric}

To quantify geometric consistency between multilevel solutions at the same timestep $t$, Zhang et al.~\shortcite{zhang2025progressive} propose to use the following metric:
\begin{align}
\label{eq:projection-op}
d_{l-1}^t = \|\Pi_{l-1}^{l}(x_l^t) - x_{l-1}^t\|_{M_{l-1}}^2,
\end{align}
where \(\Pi_{l-1}^{l}(\cdot)\) is the projection operator mapping any intermediate level ($l > 0$) geometry \(x_l^t\) to the next coarser level $l-1$. In the shell setting, this is defined as $\Pi_{l-1}^{l}(\cdot) = \big( (U_{l}^{l-1})^\top (U_{l}^{l-1}) \big)^{-1} (U_{l}^{l-1})^\top$, where $U_{l}^{l-1}$ is the linear intrinsic part of the prolongation operator $P_{l}^{l-1}(\cdot)$. In our volumetric case, $U_{l}^{l-1} = P_{l}^{l-1}$, since the prolongation $P_{l}^{l-1}$ is already linear.






    
    


\bibliographystyle{ACM-Reference-Format}
\bibliography{references}


\begin{thebibliography}{70}


\ifx \showCODEN    \undefined \def \showCODEN     #1{\unskip}     \fi
\ifx \showDOI      \undefined \def \showDOI       #1{#1}\fi
\ifx \showISBNx    \undefined \def \showISBNx     #1{\unskip}     \fi
\ifx \showISBNxiii \undefined \def \showISBNxiii  #1{\unskip}     \fi
\ifx \showISSN     \undefined \def \showISSN      #1{\unskip}     \fi
\ifx \showLCCN     \undefined \def \showLCCN      #1{\unskip}     \fi
\ifx \shownote     \undefined \def \shownote      #1{#1}          \fi
\ifx \showarticletitle \undefined \def \showarticletitle #1{#1}   \fi
\ifx \showURL      \undefined \def \showURL       {\relax}        \fi
\providecommand\bibfield[2]{#2}
\providecommand\bibinfo[2]{#2}
\providecommand\natexlab[1]{#1}
\providecommand\showeprint[2][]{arXiv:#2}

\bibitem[Aigerman and Lipman(2013)]%
        {aigerman2013injective}
\bibfield{author}{\bibinfo{person}{Noam Aigerman} {and} \bibinfo{person}{Yaron Lipman}.} \bibinfo{year}{2013}\natexlab{}.
\newblock \showarticletitle{Injective and bounded distortion mappings in 3D}.
\newblock \bibinfo{journal}{\emph{ACM Transactions on Graphics (TOG)}} \bibinfo{volume}{32}, \bibinfo{number}{4} (\bibinfo{year}{2013}), \bibinfo{pages}{1--14}.
\newblock


\bibitem[Aksoylu et~al\mbox{.}(2005)]%
        {aksoylu2005multilevel}
\bibfield{author}{\bibinfo{person}{Burak Aksoylu}, \bibinfo{person}{Andrei Khodakovsky}, {and} \bibinfo{person}{Peter Schr{\"o}der}.} \bibinfo{year}{2005}\natexlab{}.
\newblock \showarticletitle{Multilevel solvers for unstructured surface meshes}.
\newblock \bibinfo{journal}{\emph{SIAM Journal on Scientific Computing}} \bibinfo{volume}{26}, \bibinfo{number}{4} (\bibinfo{year}{2005}), \bibinfo{pages}{1146--1165}.
\newblock


\bibitem[Bergou et~al\mbox{.}(2007)]%
        {bergou2007tracks}
\bibfield{author}{\bibinfo{person}{Mikl{\'o}s Bergou}, \bibinfo{person}{Saurabh Mathur}, \bibinfo{person}{Max Wardetzky}, {and} \bibinfo{person}{Eitan Grinspun}.} \bibinfo{year}{2007}\natexlab{}.
\newblock \showarticletitle{TRACKS: Toward Directable Thin Shells}.
\newblock \bibinfo{journal}{\emph{ACM Transactions on Graphics (TOG)}} \bibinfo{volume}{26}, \bibinfo{number}{3} (\bibinfo{year}{2007}), \bibinfo{pages}{50--es}.
\newblock


\bibitem[Bolz et~al\mbox{.}(2003)]%
        {bolz_sparse_2003}
\bibfield{author}{\bibinfo{person}{J. Bolz}, \bibinfo{person}{I. Farmer}, \bibinfo{person}{E. Grinspun}, {and} \bibinfo{person}{P. Schroeder}.} \bibinfo{year}{2003}\natexlab{}.
\newblock \showarticletitle{Sparse {Matrix} {Solvers} on the {GPU}: {Conjugate} {Gradients} and {Multigrid}}.
\newblock \bibinfo{journal}{\emph{ACM Transactions on Graphics}} (\bibinfo{year}{2003}).
\newblock


\bibitem[Brandt(1986)]%
        {brandt_algebraic_1986}
\bibfield{author}{\bibinfo{person}{Achi Brandt}.} \bibinfo{year}{1986}\natexlab{}.
\newblock \showarticletitle{Algebraic multigrid theory: {The} symmetric case}.
\newblock \bibinfo{journal}{\emph{Appl. Math. Comput.}} \bibinfo{volume}{19}, \bibinfo{number}{1-4} (\bibinfo{year}{1986}), \bibinfo{pages}{23--56}.
\newblock


\bibitem[Brandt et~al\mbox{.}(2018)]%
        {10.1145/3197517.3201387}
\bibfield{author}{\bibinfo{person}{Christopher Brandt}, \bibinfo{person}{Elmar Eisemann}, {and} \bibinfo{person}{Klaus Hildebrandt}.} \bibinfo{year}{2018}\natexlab{}.
\newblock \showarticletitle{Hyper-reduced projective dynamics}.
\newblock \bibinfo{journal}{\emph{ACM Trans. Graph.}} \bibinfo{volume}{37}, \bibinfo{number}{4}, Article \bibinfo{articleno}{80} (\bibinfo{date}{July} \bibinfo{year}{2018}), \bibinfo{numpages}{13}~pages.
\newblock
\showISSN{0730-0301}
\urldef\tempurl%
\url{https://doi.org/10.1145/3197517.3201387}
\showDOI{\tempurl}


\bibitem[Campen et~al\mbox{.}(2016)]%
        {campen2016bijective}
\bibfield{author}{\bibinfo{person}{Marcel Campen}, \bibinfo{person}{Cl{\'a}udio~T Silva}, {and} \bibinfo{person}{Denis Zorin}.} \bibinfo{year}{2016}\natexlab{}.
\newblock \showarticletitle{Bijective maps from simplicial foliations}.
\newblock \bibinfo{journal}{\emph{ACM Transactions on Graphics (TOG)}} \bibinfo{volume}{35}, \bibinfo{number}{4} (\bibinfo{year}{2016}), \bibinfo{pages}{1--15}.
\newblock


\bibitem[Capell et~al\mbox{.}(2002a)]%
        {Capell:2002:ISDDD}
\bibfield{author}{\bibinfo{person}{Steve Capell}, \bibinfo{person}{Seth Green}, \bibinfo{person}{Brian Curless}, \bibinfo{person}{Tom Duchamp}, {and} \bibinfo{person}{Zoran Popovi\'{c}}.} \bibinfo{year}{2002}\natexlab{a}.
\newblock \showarticletitle{Interactive skeleton-driven dynamic deformations}.
\newblock \bibinfo{journal}{\emph{ACM Trans. Graph.}} \bibinfo{volume}{21}, \bibinfo{number}{3} (\bibinfo{date}{July} \bibinfo{year}{2002}), \bibinfo{pages}{586–593}.
\newblock
\showISSN{0730-0301}
\urldef\tempurl%
\url{https://doi.org/10.1145/566654.566622}
\showDOI{\tempurl}


\bibitem[Capell et~al\mbox{.}(2002b)]%
        {capell2002multiresolution}
\bibfield{author}{\bibinfo{person}{Steve Capell}, \bibinfo{person}{Seth Green}, \bibinfo{person}{Brian Curless}, \bibinfo{person}{Tom Duchamp}, {and} \bibinfo{person}{Zoran Popovi{\'c}}.} \bibinfo{year}{2002}\natexlab{b}.
\newblock \showarticletitle{A multiresolution framework for dynamic deformations}. In \bibinfo{booktitle}{\emph{Proceedings of the 2002 ACM SIGGRAPH/Eurographics symposium on Computer animation}}. \bibinfo{pages}{41--47}.
\newblock


\bibitem[Cherchi and Livesu(2023)]%
        {cherchi2023volmap}
\bibfield{author}{\bibinfo{person}{Gianmarco Cherchi} {and} \bibinfo{person}{Marco Livesu}.} \bibinfo{year}{2023}\natexlab{}.
\newblock \showarticletitle{VOLMAP: a Large Scale Benchmark for Volume Mappings to Simple Base Domains}. In \bibinfo{booktitle}{\emph{Computer Graphics Forum}}, Vol.~\bibinfo{volume}{42}. Wiley Online Library, \bibinfo{pages}{e14915}.
\newblock


\bibitem[Cignoni et~al\mbox{.}(2000)]%
        {cignoni2000simplification}
\bibfield{author}{\bibinfo{person}{Paolo Cignoni}, \bibinfo{person}{D Costanza}, \bibinfo{person}{Claudio Montani}, \bibinfo{person}{Claudio Rocchini}, {and} \bibinfo{person}{Roberto Scopigno}.} \bibinfo{year}{2000}\natexlab{}.
\newblock \showarticletitle{Simplification of tetrahedral meshes with accurate error evaluation}. In \bibinfo{booktitle}{\emph{Proceedings Visualization 2000. VIS 2000 (Cat. No. 00CH37145)}}. IEEE, \bibinfo{pages}{85--92}.
\newblock


\bibitem[Danovaro et~al\mbox{.}(2002)]%
        {danovaro2002multiresolution}
\bibfield{author}{\bibinfo{person}{Emanuele Danovaro}, \bibinfo{person}{Leila De~Floriani}, \bibinfo{person}{Michael Lee}, {and} \bibinfo{person}{Hanan Samet}.} \bibinfo{year}{2002}\natexlab{}.
\newblock \showarticletitle{Multiresolution tetrahedral meshes: an analysis and a comparison}. In \bibinfo{booktitle}{\emph{Proceedings SMI. Shape Modeling International 2002}}. IEEE, \bibinfo{pages}{83--273}.
\newblock


\bibitem[Debunne et~al\mbox{.}(2001)]%
        {Debunne2001}
\bibfield{author}{\bibinfo{person}{Gilles Debunne}, \bibinfo{person}{Mathieu Desbrun}, \bibinfo{person}{Marie-Paule Cani}, {and} \bibinfo{person}{Alan~H. Barr}.} \bibinfo{year}{2001}\natexlab{}.
\newblock \showarticletitle{Dynamic Real-Time Deformations using Space and Time Adaptive Sampling}. In \bibinfo{booktitle}{\emph{Proceedings of the 28th Annual Conference on Computer Graphics and Interactive Techniques (SIGGRAPH 2001)}}. \bibinfo{pages}{31--36}.
\newblock
\urldef\tempurl%
\url{https://doi.org/10.1145/383259.383262}
\showDOI{\tempurl}


\bibitem[DeRose et~al\mbox{.}(1998)]%
        {derose1998subdivision}
\bibfield{author}{\bibinfo{person}{Tony DeRose}, \bibinfo{person}{Michael Kass}, {and} \bibinfo{person}{Tien Truong}.} \bibinfo{year}{1998}\natexlab{}.
\newblock \showarticletitle{Subdivision surfaces in character animation}. In \bibinfo{booktitle}{\emph{Proceedings of the 25th annual conference on Computer graphics and interactive techniques}}. ACM, \bibinfo{pages}{85--94}.
\newblock


\bibitem[Faloutsos et~al\mbox{.}(1997)]%
        {faloutsos_dynamic_1997}
\bibfield{author}{\bibinfo{person}{P. Faloutsos}, \bibinfo{person}{M.~van~de Panne}, {and} \bibinfo{person}{D. Terzopoulos}.} \bibinfo{year}{1997}\natexlab{}.
\newblock \showarticletitle{Dynamic {Free}-{Form} {Deformations} for {Animation} {Synthesis}}.
\newblock \bibinfo{journal}{\emph{IEEE Trans. on Vis. and Comput. Graph.}} \bibinfo{volume}{3}, \bibinfo{number}{3} (\bibinfo{year}{1997}), \bibinfo{pages}{201--214}.
\newblock


\bibitem[Ferguson et~al\mbox{.}(2023)]%
        {ferguson2023timestep}
\bibfield{author}{\bibinfo{person}{Zachary Ferguson}, \bibinfo{person}{Teseo Schneider}, \bibinfo{person}{Danny~M Kaufman}, {and} \bibinfo{person}{Daniele Panozzo}.} \bibinfo{year}{2023}\natexlab{}.
\newblock \showarticletitle{In-Timestep Remeshing for Contacting Elastodynamics.}
\newblock \bibinfo{journal}{\emph{ACM Trans. Graph.}} \bibinfo{volume}{42}, \bibinfo{number}{4} (\bibinfo{year}{2023}), \bibinfo{pages}{145--1}.
\newblock


\bibitem[Floater(2003)]%
        {floater2003mean}
\bibfield{author}{\bibinfo{person}{Michael~S Floater}.} \bibinfo{year}{2003}\natexlab{}.
\newblock \showarticletitle{Mean value coordinates}.
\newblock \bibinfo{journal}{\emph{Computer Aided Geometric Design}} \bibinfo{volume}{20}, \bibinfo{number}{1} (\bibinfo{year}{2003}), \bibinfo{pages}{19--27}.
\newblock


\bibitem[Gao et~al\mbox{.}(2020)]%
        {Gao2020}
\bibfield{author}{\bibinfo{person}{Jun Gao}, \bibinfo{person}{Wenzheng Chen}, \bibinfo{person}{Tommy Xiang}, \bibinfo{person}{Clement~Fuji Tsang}, \bibinfo{person}{Alec Jacobson}, \bibinfo{person}{Morgan McGuire}, {and} \bibinfo{person}{Sanja Fidler}.} \bibinfo{year}{2020}\natexlab{}.
\newblock \showarticletitle{Learning Deformable Tetrahedral Meshes for 3D Reconstruction}. In \bibinfo{booktitle}{\emph{Advances in Neural Information Processing Systems (NeurIPS 2020)}}.
\newblock
\urldef\tempurl%
\url{https://arxiv.org/abs/2011.01437}
\showURL{%
\tempurl}


\bibitem[Garland and Heckbert(1997)]%
        {garland1997surface}
\bibfield{author}{\bibinfo{person}{Michael Garland} {and} \bibinfo{person}{Paul~S Heckbert}.} \bibinfo{year}{1997}\natexlab{}.
\newblock \showarticletitle{Surface simplification using quadric error metrics}. In \bibinfo{booktitle}{\emph{Proceedings of the 24th annual conference on Computer graphics and interactive techniques}}. \bibinfo{pages}{209--216}.
\newblock


\bibitem[Georgii and Westermann(2006)]%
        {Georgii:2006:MG4RTDefo}
\bibfield{author}{\bibinfo{person}{Joachim Georgii} {and} \bibinfo{person}{R\"{u}diger Westermann}.} \bibinfo{year}{2006}\natexlab{}.
\newblock \showarticletitle{A multigrid framework for real-time simulation of deformable bodies}.
\newblock \bibinfo{journal}{\emph{Comput. Graph.}} \bibinfo{volume}{30}, \bibinfo{number}{3} (\bibinfo{date}{June} \bibinfo{year}{2006}), \bibinfo{pages}{408–415}.
\newblock
\showISSN{0097-8493}
\urldef\tempurl%
\url{https://doi.org/10.1016/j.cag.2006.02.016}
\showDOI{\tempurl}


\bibitem[Grinspun et~al\mbox{.}(2002)]%
        {grinspun2002charms}
\bibfield{author}{\bibinfo{person}{Eitan Grinspun}, \bibinfo{person}{Petr Krysl}, {and} \bibinfo{person}{Peter Schr{\"o}der}.} \bibinfo{year}{2002}\natexlab{}.
\newblock \showarticletitle{CHARMS: A simple framework for adaptive simulation}.
\newblock \bibinfo{journal}{\emph{ACM transactions on graphics (TOG)}} \bibinfo{volume}{21}, \bibinfo{number}{3} (\bibinfo{year}{2002}), \bibinfo{pages}{281--290}.
\newblock


\bibitem[Guennebaud et~al\mbox{.}(2010)]%
        {eigenweb}
\bibfield{author}{\bibinfo{person}{Ga{\"e}l Guennebaud}, \bibinfo{person}{Beno{\^i}t Jacob}, {et~al\mbox{.}}} \bibinfo{year}{2010}\natexlab{}.
\newblock \bibinfo{title}{Eigen v3}.
\newblock
\newblock


\bibitem[Hang(2015)]%
        {hang2015tetgen}
\bibfield{author}{\bibinfo{person}{Si Hang}.} \bibinfo{year}{2015}\natexlab{}.
\newblock \showarticletitle{TetGen, a Delaunay-based quality tetrahedral mesh generator}.
\newblock \bibinfo{journal}{\emph{ACM Trans. Math. Softw}} \bibinfo{volume}{41}, \bibinfo{number}{2} (\bibinfo{year}{2015}), \bibinfo{pages}{11}.
\newblock


\bibitem[Hinderink et~al\mbox{.}(2024)]%
        {hinderink2024bijective}
\bibfield{author}{\bibinfo{person}{Steffen Hinderink}, \bibinfo{person}{Hendrik Br{\"u}ckler}, {and} \bibinfo{person}{Marcel Campen}.} \bibinfo{year}{2024}\natexlab{}.
\newblock \showarticletitle{Bijective volumetric mapping via star decomposition}.
\newblock \bibinfo{journal}{\emph{ACM Transactions on Graphics (TOG)}} \bibinfo{volume}{43}, \bibinfo{number}{6} (\bibinfo{year}{2024}), \bibinfo{pages}{1--11}.
\newblock


\bibitem[Hinderink and Campen(2023)]%
        {hinderink2023galaxy}
\bibfield{author}{\bibinfo{person}{Steffen Hinderink} {and} \bibinfo{person}{Marcel Campen}.} \bibinfo{year}{2023}\natexlab{}.
\newblock \showarticletitle{Galaxy maps: Localized foliations for bijective volumetric mapping}.
\newblock \bibinfo{journal}{\emph{ACM Transactions on Graphics (TOG)}} \bibinfo{volume}{42}, \bibinfo{number}{4} (\bibinfo{year}{2023}), \bibinfo{pages}{1--16}.
\newblock


\bibitem[Hoppe(1996)]%
        {Hoppe1996}
\bibfield{author}{\bibinfo{person}{Hugues Hoppe}.} \bibinfo{year}{1996}\natexlab{}.
\newblock \showarticletitle{Progressive Meshes}. In \bibinfo{booktitle}{\emph{Proceedings of the 23rd Annual Conference on Computer Graphics and Interactive Techniques}} \emph{(\bibinfo{series}{SIGGRAPH '96})}. \bibinfo{publisher}{Association for Computing Machinery}, \bibinfo{address}{New York, NY, USA}, \bibinfo{pages}{99–108}.
\newblock
\showISBNx{0897917464}


\bibitem[Hu et~al\mbox{.}(2020)]%
        {hu2020fast}
\bibfield{author}{\bibinfo{person}{Yixin Hu}, \bibinfo{person}{Teseo Schneider}, \bibinfo{person}{Bolun Wang}, \bibinfo{person}{Denis Zorin}, {and} \bibinfo{person}{Daniele Panozzo}.} \bibinfo{year}{2020}\natexlab{}.
\newblock \showarticletitle{Fast tetrahedral meshing in the wild}.
\newblock \bibinfo{journal}{\emph{ACM Transactions on Graphics (ToG)}} \bibinfo{volume}{39}, \bibinfo{number}{4} (\bibinfo{year}{2020}), \bibinfo{pages}{117--1}.
\newblock


\bibitem[Hu et~al\mbox{.}(2018)]%
        {hu2018tetrahedral}
\bibfield{author}{\bibinfo{person}{Yixin Hu}, \bibinfo{person}{Qingnan Zhou}, \bibinfo{person}{Xifeng Gao}, \bibinfo{person}{Alec Jacobson}, \bibinfo{person}{Denis Zorin}, {and} \bibinfo{person}{Daniele Panozzo}.} \bibinfo{year}{2018}\natexlab{}.
\newblock \showarticletitle{Tetrahedral meshing in the wild.}
\newblock \bibinfo{journal}{\emph{ACM Trans. Graph.}} \bibinfo{volume}{37}, \bibinfo{number}{4} (\bibinfo{year}{2018}), \bibinfo{pages}{60}.
\newblock


\bibitem[Jacobson et~al\mbox{.}(2012)]%
        {10.1145/2185520.2185573}
\bibfield{author}{\bibinfo{person}{Alec Jacobson}, \bibinfo{person}{Ilya Baran}, \bibinfo{person}{Ladislav Kavan}, \bibinfo{person}{Jovan Popovi\'{c}}, {and} \bibinfo{person}{Olga Sorkine}.} \bibinfo{year}{2012}\natexlab{}.
\newblock \showarticletitle{Fast automatic skinning transformations}.
\newblock \bibinfo{journal}{\emph{ACM Trans. Graph.}} \bibinfo{volume}{31}, \bibinfo{number}{4}, Article \bibinfo{articleno}{77} (\bibinfo{date}{July} \bibinfo{year}{2012}), \bibinfo{numpages}{10}~pages.
\newblock
\showISSN{0730-0301}
\urldef\tempurl%
\url{https://doi.org/10.1145/2185520.2185573}
\showDOI{\tempurl}


\bibitem[Jacobson et~al\mbox{.}(2011)]%
        {Jacobson2011}
\bibfield{author}{\bibinfo{person}{Alec Jacobson}, \bibinfo{person}{Emre~S. Tosun}, \bibinfo{person}{Olga Sorkine}, {and} \bibinfo{person}{Denis Zorin}.} \bibinfo{year}{2011}\natexlab{}.
\newblock \showarticletitle{Bounded biharmonic weights for real-time deformation}.
\newblock \bibinfo{journal}{\emph{ACM Transactions on Graphics (TOG)}} \bibinfo{volume}{30}, \bibinfo{number}{4} (\bibinfo{year}{2011}), \bibinfo{pages}{78}.
\newblock
\urldef\tempurl%
\url{https://doi.org/10.1145/2010324.1964973}
\showDOI{\tempurl}


\bibitem[James(2020)]%
        {James2021}
\bibfield{author}{\bibinfo{person}{Doug~L. James}.} \bibinfo{year}{2020}\natexlab{}.
\newblock \showarticletitle{Phong Deformation: A Better C0 Interpolant for Embedded Deformation}.
\newblock \bibinfo{journal}{\emph{ACM Transactions on Graphics (TOG)}} \bibinfo{volume}{39}, \bibinfo{number}{4} (\bibinfo{year}{2020}), \bibinfo{pages}{126:1--126:12}.
\newblock
\urldef\tempurl%
\url{https://doi.org/10.1145/3386569.3392411}
\showDOI{\tempurl}


\bibitem[James and Pai(2003)]%
        {james2003multiresolution}
\bibfield{author}{\bibinfo{person}{Doug~L James} {and} \bibinfo{person}{Dinesh~K Pai}.} \bibinfo{year}{2003}\natexlab{}.
\newblock \showarticletitle{Multiresolution Green's function methods for interactive simulation of large-scale elastostatic objects}.
\newblock \bibinfo{journal}{\emph{ACM Transactions on Graphics (TOG)}} \bibinfo{volume}{22}, \bibinfo{number}{1} (\bibinfo{year}{2003}), \bibinfo{pages}{47--82}.
\newblock


\bibitem[Joshi et~al\mbox{.}(2007)]%
        {joshi2007harmonic}
\bibfield{author}{\bibinfo{person}{Pankaj Joshi}, \bibinfo{person}{Mark Meyer}, \bibinfo{person}{Tony DeRose}, \bibinfo{person}{Brian Green}, {and} \bibinfo{person}{Thomas Sanocki}.} \bibinfo{year}{2007}\natexlab{}.
\newblock \showarticletitle{Harmonic coordinates for character articulation}. In \bibinfo{booktitle}{\emph{ACM SIGGRAPH 2007 papers}}. ACM, \bibinfo{pages}{71--es}.
\newblock


\bibitem[Ju et~al\mbox{.}(2005)]%
        {Ju2005}
\bibfield{author}{\bibinfo{person}{Tao Ju}, \bibinfo{person}{Scott Schaefer}, {and} \bibinfo{person}{Joe Warren}.} \bibinfo{year}{2005}\natexlab{}.
\newblock \showarticletitle{Mean value coordinates for closed triangular meshes}. In \bibinfo{booktitle}{\emph{ACM SIGGRAPH 2005 Papers}}. ACM, \bibinfo{pages}{561--566}.
\newblock
\urldef\tempurl%
\url{https://doi.org/10.1145/1186822.1073229}
\showDOI{\tempurl}


\bibitem[Kane et~al\mbox{.}(1999)]%
        {kane1999finite}
\bibfield{author}{\bibinfo{person}{Couro Kane}, \bibinfo{person}{Eduardo~A Repetto}, \bibinfo{person}{Michael Ortiz}, {and} \bibinfo{person}{Jerrold~E Marsden}.} \bibinfo{year}{1999}\natexlab{}.
\newblock \showarticletitle{Finite element analysis of nonsmooth contact}.
\newblock \bibinfo{journal}{\emph{CMAME}} \bibinfo{volume}{180}, \bibinfo{number}{1-2} (\bibinfo{year}{1999}).
\newblock


\bibitem[Kharevych et~al\mbox{.}(2009)]%
        {kharevych2009numerical}
\bibfield{author}{\bibinfo{person}{Lily Kharevych}, \bibinfo{person}{Patrick Mullen}, \bibinfo{person}{Houman Owhadi}, {and} \bibinfo{person}{Mathieu Desbrun}.} \bibinfo{year}{2009}\natexlab{}.
\newblock \showarticletitle{Numerical coarsening of inhomogeneous elastic materials}.
\newblock \bibinfo{journal}{\emph{ACM Transactions on graphics (TOG)}} \bibinfo{volume}{28}, \bibinfo{number}{3} (\bibinfo{year}{2009}), \bibinfo{pages}{1--8}.
\newblock


\bibitem[Kim and James(2009)]%
        {kim2009skipping}
\bibfield{author}{\bibinfo{person}{Theodore Kim} {and} \bibinfo{person}{Doug~L James}.} \bibinfo{year}{2009}\natexlab{}.
\newblock \showarticletitle{Skipping steps in deformable simulation with online model reduction}.
\newblock \bibinfo{journal}{\emph{ACM Transactions on Graphics (TOG)}} \bibinfo{volume}{28}, \bibinfo{number}{5} (\bibinfo{year}{2009}), \bibinfo{pages}{1--9}.
\newblock


\bibitem[Kovalsky et~al\mbox{.}(2014)]%
        {kovalsky2014controlling}
\bibfield{author}{\bibinfo{person}{Shahar~Z Kovalsky}, \bibinfo{person}{Noam Aigerman}, \bibinfo{person}{Ronen Basri}, {and} \bibinfo{person}{Yaron Lipman}.} \bibinfo{year}{2014}\natexlab{}.
\newblock \showarticletitle{Controlling singular values with semidefinite programming.}
\newblock \bibinfo{journal}{\emph{ACM Trans. Graph.}} \bibinfo{volume}{33}, \bibinfo{number}{4} (\bibinfo{year}{2014}), \bibinfo{pages}{68--1}.
\newblock


\bibitem[Lee et~al\mbox{.}(1998)]%
        {Lee1998}
\bibfield{author}{\bibinfo{person}{Aaron W.~F. Lee}, \bibinfo{person}{Wim Sweldens}, \bibinfo{person}{Peter Schr{\"{o}}der}, \bibinfo{person}{Lawrence~C. Cowsar}, {and} \bibinfo{person}{David~P. Dobkin}.} \bibinfo{year}{1998}\natexlab{}.
\newblock \showarticletitle{{MAPS:} Multiresolution Adaptive Parameterization of Surfaces}. In \bibinfo{booktitle}{\emph{Proceedings of the 25th Annual Conference on Computer Graphics and Interactive Techniques, {SIGGRAPH} 1998, Orlando, FL, USA, July 19-24, 1998}}, \bibfield{editor}{\bibinfo{person}{Steve Cunningham}, \bibinfo{person}{Walt Bransford}, {and} \bibinfo{person}{Michael~F. Cohen}} (Eds.). \bibinfo{publisher}{{ACM}}, \bibinfo{pages}{95--104}.
\newblock
\urldef\tempurl%
\url{https://doi.org/10.1145/280814.280828}
\showDOI{\tempurl}


\bibitem[L{\'e}ger et~al\mbox{.}(2014)]%
        {leger2014updated}
\bibfield{author}{\bibinfo{person}{S L{\'e}ger}, \bibinfo{person}{A Fortin}, \bibinfo{person}{C Tibirna}, {and} \bibinfo{person}{M Fortin}.} \bibinfo{year}{2014}\natexlab{}.
\newblock \showarticletitle{An updated Lagrangian method with error estimation and adaptive remeshing for very large deformation elasticity problems}.
\newblock \bibinfo{journal}{\emph{Internat. J. Numer. Methods Engrg.}} \bibinfo{volume}{100}, \bibinfo{number}{13} (\bibinfo{year}{2014}), \bibinfo{pages}{1006--1030}.
\newblock


\bibitem[Li et~al\mbox{.}(2020)]%
        {Li2020IPC}
\bibfield{author}{\bibinfo{person}{Minchen Li}, \bibinfo{person}{Zachary Ferguson}, \bibinfo{person}{Teseo Schneider}, \bibinfo{person}{Timothy Langlois}, \bibinfo{person}{Denis Zorin}, \bibinfo{person}{Daniele Panozzo}, \bibinfo{person}{Chenfanfu Jiang}, {and} \bibinfo{person}{Danny~M. Kaufman}.} \bibinfo{year}{2020}\natexlab{}.
\newblock \showarticletitle{Incremental Potential Contact: Intersection- and Inversion-free Large Deformation Dynamics}.
\newblock \bibinfo{journal}{\emph{ACM Trans. Graph. (SIGGRAPH)}} \bibinfo{volume}{39}, \bibinfo{number}{4}, Article \bibinfo{articleno}{49} (\bibinfo{year}{2020}).
\newblock


\bibitem[Li et~al\mbox{.}(2021)]%
        {Li2021CIPC}
\bibfield{author}{\bibinfo{person}{Minchen Li}, \bibinfo{person}{Danny~M. Kaufman}, {and} \bibinfo{person}{Chenfanfu Jiang}.} \bibinfo{year}{2021}\natexlab{}.
\newblock \showarticletitle{Codimensional Incremental Potential Contact}.
\newblock \bibinfo{journal}{\emph{ACM Trans. Graph.}} \bibinfo{volume}{40}, \bibinfo{number}{4}, Article \bibinfo{articleno}{170} (\bibinfo{date}{jul} \bibinfo{year}{2021}), \bibinfo{numpages}{24}~pages.
\newblock
\showISSN{0730-0301}


\bibitem[Liu et~al\mbox{.}(2021)]%
        {Liu:2021:SMIP}
\bibfield{author}{\bibinfo{person}{Hsueh-Ti~Derek Liu}, \bibinfo{person}{Jiayi~Eris Zhang}, \bibinfo{person}{Mirela Ben-Chen}, {and} \bibinfo{person}{Alec Jacobson}.} \bibinfo{year}{2021}\natexlab{}.
\newblock \showarticletitle{Surface Multigrid via Intrinsic Prolongation}.
\newblock \bibinfo{journal}{\emph{ACM Trans. Graph.}} \bibinfo{volume}{40}, \bibinfo{number}{4}, Article \bibinfo{articleno}{80} (\bibinfo{date}{jul} \bibinfo{year}{2021}), \bibinfo{numpages}{13}~pages.
\newblock
\showISSN{0730-0301}


\bibitem[Losasso et~al\mbox{.}(2004)]%
        {Losasso2004}
\bibfield{author}{\bibinfo{person}{Frank Losasso}, \bibinfo{person}{Fr{\'e}do Durand}, {and} \bibinfo{person}{Joe Stam}.} \bibinfo{year}{2004}\natexlab{}.
\newblock \showarticletitle{Adaptive Simulation of Smoke and Other Compressible Fluids}.
\newblock \bibinfo{journal}{\emph{ACM Transactions on Graphics}} \bibinfo{volume}{23}, \bibinfo{number}{3} (\bibinfo{year}{2004}), \bibinfo{pages}{457--462}.
\newblock
\urldef\tempurl%
\url{https://doi.org/10.1145/1015706.1015745}
\showDOI{\tempurl}


\bibitem[McAdams et~al\mbox{.}(2011)]%
        {McAdams:2011:MGCharacters}
\bibfield{author}{\bibinfo{person}{Aleka McAdams}, \bibinfo{person}{Yongning Zhu}, \bibinfo{person}{Andrew Selle}, \bibinfo{person}{Mark Empey}, \bibinfo{person}{Rasmus Tamstorf}, \bibinfo{person}{Joseph Teran}, {and} \bibinfo{person}{Eftychios Sifakis}.} \bibinfo{year}{2011}\natexlab{}.
\newblock \showarticletitle{Efficient elasticity for character skinning with contact and collisions}. In \bibinfo{booktitle}{\emph{ACM SIGGRAPH 2011 Papers}} (Vancouver, British Columbia, Canada) \emph{(\bibinfo{series}{SIGGRAPH '11})}. \bibinfo{publisher}{Association for Computing Machinery}, \bibinfo{address}{New York, NY, USA}, Article \bibinfo{articleno}{37}, \bibinfo{numpages}{12}~pages.
\newblock
\showISBNx{9781450309431}
\urldef\tempurl%
\url{https://doi.org/10.1145/1964921.1964932}
\showDOI{\tempurl}


\bibitem[Molino et~al\mbox{.}(2004)]%
        {molino_virtual_2004}
\bibfield{author}{\bibinfo{person}{Neil Molino}, \bibinfo{person}{Zhaosheng Bao}, {and} \bibinfo{person}{Ron Fedkiw}.} \bibinfo{year}{2004}\natexlab{}.
\newblock \showarticletitle{A virtual node algorithm for changing mesh topology during simulation}.
\newblock \bibinfo{journal}{\emph{ACM Transactions on Graphics}} \bibinfo{volume}{23}, \bibinfo{number}{3} (\bibinfo{date}{Aug.} \bibinfo{year}{2004}), \bibinfo{pages}{385--392}.
\newblock


\bibitem[Molino et~al\mbox{.}(2003)]%
        {Molino2003}
\bibfield{author}{\bibinfo{person}{Neil Molino}, \bibinfo{person}{Robert Bridson}, {and} \bibinfo{person}{Ronald Fedkiw}.} \bibinfo{year}{2003}\natexlab{}.
\newblock \showarticletitle{Tetrahedral Mesh Generation for Deformable Bodies}. In \bibinfo{booktitle}{\emph{Proceedings of the ACM SIGGRAPH/Eurographics Symposium on Computer Animation}}. \bibinfo{pages}{25--35}.
\newblock
\urldef\tempurl%
\url{https://graphics.stanford.edu/papers/meshing-sig03/}
\showURL{%
\tempurl}


\bibitem[Müller and Gross(2004)]%
        {muller_interactive_2004}
\bibfield{author}{\bibinfo{person}{M. Müller} {and} \bibinfo{person}{M. Gross}.} \bibinfo{year}{2004}\natexlab{}.
\newblock \showarticletitle{Interactive {Virtual} {Materials}}. In \bibinfo{booktitle}{\emph{Graph. {Interface}}}. \bibinfo{pages}{239--246}.
\newblock


\bibitem[Renze and Oliver(1996)]%
        {renze1996generalized}
\bibfield{author}{\bibinfo{person}{Kevin~J Renze} {and} \bibinfo{person}{James~H Oliver}.} \bibinfo{year}{1996}\natexlab{}.
\newblock \showarticletitle{Generalized unstructured decimation [computer graphics]}.
\newblock \bibinfo{journal}{\emph{IEEE Computer Graphics and Applications}} \bibinfo{volume}{16}, \bibinfo{number}{6} (\bibinfo{year}{1996}), \bibinfo{pages}{24--32}.
\newblock


\bibitem[Rivers and James(2007)]%
        {rivers_fastlsm_2007}
\bibfield{author}{\bibinfo{person}{Alec~R. Rivers} {and} \bibinfo{person}{Doug~L. James}.} \bibinfo{year}{2007}\natexlab{}.
\newblock \showarticletitle{{FastLSM}: fast lattice shape matching for robust real-time deformation}.
\newblock \bibinfo{journal}{\emph{ACM Trans. Graph.}} \bibinfo{volume}{26}, \bibinfo{number}{3} (\bibinfo{date}{July} \bibinfo{year}{2007}), \bibinfo{pages}{82--es}.
\newblock
\showISSN{0730-0301}
\urldef\tempurl%
\url{https://doi.org/10.1145/1276377.1276480}
\showDOI{\tempurl}
\newblock
\shownote{Place: New York, NY, USA Publisher: Association for Computing Machinery}.


\bibitem[Ruan et~al\mbox{.}(2024)]%
        {ruan2024minnie}
\bibfield{author}{\bibinfo{person}{Liangwang Ruan}, \bibinfo{person}{Bin Wang}, \bibinfo{person}{Tiantian Liu}, {and} \bibinfo{person}{Baoquan Chen}.} \bibinfo{year}{2024}\natexlab{}.
\newblock \showarticletitle{MiNNIE: a Mixed Multigrid Method for Real-time Simulation of Nonlinear Near-Incompressible Elastics}.
\newblock \bibinfo{journal}{\emph{ACM Transactions on Graphics (TOG)}} \bibinfo{volume}{43}, \bibinfo{number}{6} (\bibinfo{year}{2024}), \bibinfo{pages}{1--15}.
\newblock


\bibitem[Schneider(2017)]%
        {schneider2017theory}
\bibfield{author}{\bibinfo{person}{Teseo Schneider}.} \bibinfo{year}{2017}\natexlab{}.
\newblock \showarticletitle{Theory and applications of bijective barycentric mappings}.
\newblock  (\bibinfo{year}{2017}).
\newblock


\bibitem[Shao et~al\mbox{.}(2022)]%
        {10.1145/3528223.3530109}
\bibfield{author}{\bibinfo{person}{Han Shao}, \bibinfo{person}{Libo Huang}, {and} \bibinfo{person}{Dominik~L. Michels}.} \bibinfo{year}{2022}\natexlab{}.
\newblock \showarticletitle{A fast unsmoothed aggregation algebraic multigrid framework for the large-scale simulation of incompressible flow}.
\newblock \bibinfo{journal}{\emph{ACM Trans. Graph.}} \bibinfo{volume}{41}, \bibinfo{number}{4}, Article \bibinfo{articleno}{49} (\bibinfo{date}{July} \bibinfo{year}{2022}), \bibinfo{numpages}{18}~pages.
\newblock
\showISSN{0730-0301}
\urldef\tempurl%
\url{https://doi.org/10.1145/3528223.3530109}
\showDOI{\tempurl}


\bibitem[Staadt and Gross(1998)]%
        {staadt1998progressive}
\bibfield{author}{\bibinfo{person}{Oliver~G Staadt} {and} \bibinfo{person}{Markus~H Gross}.} \bibinfo{year}{1998}\natexlab{}.
\newblock \bibinfo{booktitle}{\emph{Progressive tetrahedralizations}}.
\newblock \bibinfo{publisher}{IEEE}.
\newblock


\bibitem[Teng et~al\mbox{.}(2015)]%
        {teng2015subspace}
\bibfield{author}{\bibinfo{person}{Yun Teng}, \bibinfo{person}{Mark Meyer}, \bibinfo{person}{Tony DeRose}, {and} \bibinfo{person}{Theodore Kim}.} \bibinfo{year}{2015}\natexlab{}.
\newblock \showarticletitle{Subspace condensation: Full space adaptivity for subspace deformations}.
\newblock \bibinfo{journal}{\emph{ACM Transactions on Graphics (TOG)}} \bibinfo{volume}{34}, \bibinfo{number}{4} (\bibinfo{year}{2015}), \bibinfo{pages}{1--9}.
\newblock


\bibitem[Trettner and Kobbelt(2020)]%
        {10.1111:cgf.13933}
\bibfield{author}{\bibinfo{person}{Philip Trettner} {and} \bibinfo{person}{Leif Kobbelt}.} \bibinfo{year}{2020}\natexlab{}.
\newblock \showarticletitle{{Fast and Robust QEF Minimization using Probabilistic Quadrics}}.
\newblock \bibinfo{journal}{\emph{Computer Graphics Forum}} (\bibinfo{year}{2020}).
\newblock
\showISSN{1467-8659}
\urldef\tempurl%
\url{https://doi.org/10.1111/cgf.13933}
\showDOI{\tempurl}


\bibitem[Trottenberg et~al\mbox{.}(2001)]%
        {trottenberg_multigrid_2001}
\bibfield{author}{\bibinfo{person}{Ulrich Trottenberg}, \bibinfo{person}{Cornelius~W. Oosterlee}, {and} \bibinfo{person}{Anton Schuller}.} \bibinfo{year}{2001}\natexlab{}.
\newblock \bibinfo{booktitle}{\emph{Multigrid}}.
\newblock \bibinfo{publisher}{Academic Press}.
\newblock


\bibitem[Trotts et~al\mbox{.}(1999)]%
        {trotts1999simplification}
\bibfield{author}{\bibinfo{person}{Issac~J Trotts}, \bibinfo{person}{Bernd Hamann}, {and} \bibinfo{person}{Kenneth~I Joy}.} \bibinfo{year}{1999}\natexlab{}.
\newblock \showarticletitle{Simplification of tetrahedral meshes with error bounds}.
\newblock \bibinfo{journal}{\emph{IEEE Transactions on Visualization and Computer Graphics}} \bibinfo{volume}{5}, \bibinfo{number}{3} (\bibinfo{year}{1999}), \bibinfo{pages}{224--237}.
\newblock


\bibitem[Trusty et~al\mbox{.}(2024)]%
        {trusty2024trading}
\bibfield{author}{\bibinfo{person}{Ty Trusty}, \bibinfo{person}{Yun~(Raymond) Fei}, \bibinfo{person}{David~I.W. Levin}, {and} \bibinfo{person}{Danny~M. Kaufman}.} \bibinfo{year}{2024}\natexlab{}.
\newblock \showarticletitle{Trading Spaces: Adaptive Subspace Time Integration for Contacting Elastodynamics}.
\newblock \bibinfo{journal}{\emph{ACM Transactions on Graphics}} \bibinfo{volume}{43}, \bibinfo{number}{6} (\bibinfo{year}{2024}), \bibinfo{pages}{1--16}.
\newblock
\urldef\tempurl%
\url{https://doi.org/10.1145/3687946}
\showDOI{\tempurl}


\bibitem[Tutte(1963)]%
        {Tutte1963}
\bibfield{author}{\bibinfo{person}{William~T. Tutte}.} \bibinfo{year}{1963}\natexlab{}.
\newblock \showarticletitle{How to Draw a Graph}.
\newblock \bibinfo{journal}{\emph{Proceedings of the London Mathematical Society}} \bibinfo{volume}{13}, \bibinfo{number}{3} (\bibinfo{year}{1963}), \bibinfo{pages}{743--767}.
\newblock
\urldef\tempurl%
\url{https://doi.org/10.1112/plms/s3-13.1.743}
\showDOI{\tempurl}


\bibitem[Vavourakis et~al\mbox{.}(2013)]%
        {vavourakis2013assessment}
\bibfield{author}{\bibinfo{person}{Vasileios Vavourakis}, \bibinfo{person}{Dimitrios Loukidis}, \bibinfo{person}{Dimos~C Charmpis}, {and} \bibinfo{person}{Panos Papanastasiou}.} \bibinfo{year}{2013}\natexlab{}.
\newblock \showarticletitle{Assessment of remeshing and remapping strategies for large deformation elastoplastic finite element analysis}.
\newblock \bibinfo{journal}{\emph{Computers \& Structures}}  \bibinfo{volume}{114} (\bibinfo{year}{2013}), \bibinfo{pages}{133--146}.
\newblock


\bibitem[Wang et~al\mbox{.}(2017)]%
        {wang2017cnn}
\bibfield{author}{\bibinfo{person}{Peng-Shuai Wang}, \bibinfo{person}{Yang Liu}, \bibinfo{person}{Yu-Xiao Guo}, \bibinfo{person}{Chun-Yu Sun}, {and} \bibinfo{person}{Xin Tong}.} \bibinfo{year}{2017}\natexlab{}.
\newblock \showarticletitle{O-cnn: Octree-based convolutional neural networks for 3d shape analysis}.
\newblock \bibinfo{journal}{\emph{ACM Transactions On Graphics (TOG)}} \bibinfo{volume}{36}, \bibinfo{number}{4} (\bibinfo{year}{2017}), \bibinfo{pages}{1--11}.
\newblock


\bibitem[Wang et~al\mbox{.}(2015)]%
        {wang2015linear}
\bibfield{author}{\bibinfo{person}{Yu Wang}, \bibinfo{person}{Alec Jacobson}, \bibinfo{person}{Jernej Barbi{\v{c}}}, {and} \bibinfo{person}{Ladislav Kavan}.} \bibinfo{year}{2015}\natexlab{}.
\newblock \showarticletitle{Linear subspace design for real-time shape deformation}.
\newblock \bibinfo{journal}{\emph{ACM Transactions on Graphics (TOG)}} \bibinfo{volume}{34}, \bibinfo{number}{4} (\bibinfo{year}{2015}), \bibinfo{pages}{1--11}.
\newblock


\bibitem[Xian et~al\mbox{.}(2019)]%
        {xian_scalable_2019}
\bibfield{author}{\bibinfo{person}{Zangyueyang Xian}, \bibinfo{person}{Xin Tong}, {and} \bibinfo{person}{Tiantian Liu}.} \bibinfo{year}{2019}\natexlab{}.
\newblock \showarticletitle{A {Scalable} {Galerkin} {Multigrid} {Method} for {Real}-time {Simulation} of {Deformable} {Objects}}.
\newblock \bibinfo{journal}{\emph{ACM Trans. Graph. (TOG)}} \bibinfo{volume}{38}, \bibinfo{number}{6} (\bibinfo{year}{2019}).
\newblock


\bibitem[Zhang et~al\mbox{.}(2023)]%
        {zhang2023progressive}
\bibfield{author}{\bibinfo{person}{Jiayi~Eris Zhang}, \bibinfo{person}{J{\'e}r{\'e}mie Dumas}, \bibinfo{person}{Yun Fei}, \bibinfo{person}{Alec Jacobson}, \bibinfo{person}{Doug~L James}, {and} \bibinfo{person}{Danny~M Kaufman}.} \bibinfo{year}{2023}\natexlab{}.
\newblock \showarticletitle{Progressive Shell Quasistatics for Unstructured Meshes}.
\newblock \bibinfo{journal}{\emph{ACM Transactions on Graphics (TOG)}} \bibinfo{volume}{42}, \bibinfo{number}{6} (\bibinfo{year}{2023}), \bibinfo{pages}{1--17}.
\newblock


\bibitem[Zhang et~al\mbox{.}(2022)]%
        {10.1145/3550454.3555510}
\bibfield{author}{\bibinfo{person}{Jiayi~Eris Zhang}, \bibinfo{person}{J\'{e}r\'{e}mie Dumas}, \bibinfo{person}{Yun~(Raymond) Fei}, \bibinfo{person}{Alec Jacobson}, \bibinfo{person}{Doug~L. James}, {and} \bibinfo{person}{Danny~M. Kaufman}.} \bibinfo{year}{2022}\natexlab{}.
\newblock \showarticletitle{Progressive Simulation for Cloth Quasistatics}.
\newblock \bibinfo{journal}{\emph{ACM Trans. Graph.}} \bibinfo{volume}{41}, \bibinfo{number}{6}, Article \bibinfo{articleno}{218} (\bibinfo{date}{nov} \bibinfo{year}{2022}), \bibinfo{numpages}{16}~pages.
\newblock
\showISSN{0730-0301}
\urldef\tempurl%
\url{https://doi.org/10.1145/3550454.3555510}
\showDOI{\tempurl}


\bibitem[Zhang et~al\mbox{.}(2024)]%
        {zhang2024progressive}
\bibfield{author}{\bibinfo{person}{Jiayi~Eris Zhang}, \bibinfo{person}{Doug James}, {and} \bibinfo{person}{Danny~M Kaufman}.} \bibinfo{year}{2024}\natexlab{}.
\newblock \showarticletitle{Progressive Dynamics for Cloth and Shell Animation}.
\newblock \bibinfo{journal}{\emph{ACM Transactions on Graphics (TOG)}} \bibinfo{volume}{43}, \bibinfo{number}{4} (\bibinfo{year}{2024}), \bibinfo{pages}{1--18}.
\newblock


\bibitem[Zhang et~al\mbox{.}(2025)]%
        {zhang2025progressive}
\bibfield{author}{\bibinfo{person}{Jiayi~Eris Zhang}, \bibinfo{person}{Doug James}, {and} \bibinfo{person}{Danny~M Kaufman}.} \bibinfo{year}{2025}\natexlab{}.
\newblock \showarticletitle{Progressive Dynamics++: A Framework for Stable, Continuous, and Consistent Animation Across Resolution and Time}.
\newblock \bibinfo{journal}{\emph{ACM Transactions on Graphics (TOG)}} \bibinfo{volume}{44}, \bibinfo{number}{4} (\bibinfo{year}{2025}), \bibinfo{pages}{1--20}.
\newblock


\bibitem[Zheng and James(2011)]%
        {zheng2011toward}
\bibfield{author}{\bibinfo{person}{Changxi Zheng} {and} \bibinfo{person}{Doug~L James}.} \bibinfo{year}{2011}\natexlab{}.
\newblock \showarticletitle{Toward high-quality modal contact sound}.
\newblock \bibinfo{journal}{\emph{ACM Transactions on Graphics (TOG)}} \bibinfo{volume}{30}, \bibinfo{number}{4} (\bibinfo{year}{2011}), \bibinfo{pages}{1--12}.
\newblock


\bibitem[Zhu et~al\mbox{.}(2010)]%
        {zhu2010efficient}
\bibfield{author}{\bibinfo{person}{Yongning Zhu}, \bibinfo{person}{Eftychios Sifakis}, \bibinfo{person}{Joseph Teran}, {and} \bibinfo{person}{Achi Brandt}.} \bibinfo{year}{2010}\natexlab{}.
\newblock \showarticletitle{An efficient multigrid method for the simulation of high-resolution elastic solids}.
\newblock \bibinfo{journal}{\emph{ACM Transactions on Graphics (TOG)}} \bibinfo{volume}{29}, \bibinfo{number}{2} (\bibinfo{year}{2010}), \bibinfo{pages}{1--18}.
\newblock


\end{thebibliography}

\end{document}